\documentclass[11pt]{article}

\usepackage[utf8]{inputenc}

\usepackage[ruled,vlined,linesnumbered,noend]{algorithm2e}
\SetKwInOut{Parameter}{Parameters}

\usepackage{amsthm}
\usepackage{amsmath}
\usepackage{amssymb}
\usepackage{enumerate}
\usepackage{graphicx}
\usepackage[margin=1in]{geometry}
\usepackage{dsfont}
\usepackage{booktabs} 
\usepackage{tabularx}
\usepackage{multirow}
\usepackage{enumitem}
\usepackage{tablefootnote}

\usepackage{soul}

\usepackage[dvipsnames]{xcolor}
\definecolor{ForestGreen}{rgb}{0.1333,0.5451,0.1333}
\definecolor{DarkRed}{rgb}{0.65,0,0}

\usepackage{hyperref}
\hypersetup{
    colorlinks=true,
    linkcolor=DarkRed,
    filecolor=magenta,      
    urlcolor=cyan,
    citecolor=ForestGreen,
}

\usepackage{cleveref}

\usepackage{framed}
\usepackage[framemethod=tikz]{mdframed}
\usepackage{tikz-cd}

\newenvironment{wrapper}[1]
{
	\begin{center}
		\begin{minipage}{\linewidth}
			\begin{mdframed}[hidealllines=true, backgroundcolor=gray!20, leftmargin=0cm,innerleftmargin=0.4cm,innerrightmargin=0.4cm,innertopmargin=0.4cm,innerbottommargin=0.4cm,roundcorner=0pt]
				#1}
			{\end{mdframed}
		\end{minipage}
	\end{center}
}

\DeclareMathAlphabet{\mathmybb}{U}{bbold}{m}{n}

\newtheorem{theorem}{Theorem}[section]
\newtheorem{lemma}[theorem]{Lemma}

\newtheorem{corollary}[theorem]{Corollary}

\newtheorem{claim}[theorem]{Claim}
\newtheorem{question}{Question}

\DeclareMathOperator*{\poly}{poly}

\newcommand{\proj}{\pi}

\newcommand{\A}{\textsc{Alg}}

\newcommand{\OPT}{\textsc{OPT}}

\newcommand{\cl}{\texttt{cost}}

\newcommand{\cost}{\textnormal{\texttt{cost}}}
\newcommand{\dist}{\textnormal{\texttt{dist}}}

\newcommand{\MPAlg}{\textnormal{\textsc{MP-Alg}}}
\newcommand{\Restr}{\textnormal{\textsc{Res-Greedy}}}

\newcommand{\calP}[0]{\mathcal{P}}

\newcommand{\adv}[0]{\mathcal{A}}

\title{Deterministic $k$-Median Clustering in Near-Optimal Time}

\author{Mart\'{i}n Costa\thanks{University of Warwick, \texttt{Martin.Costa@warwick.ac.uk}. Supported by a Google PhD Fellowship.}
\and Ermiya Farokhnejad\thanks{University of Warwick, \texttt{Ermiya.Farokhnejad@warwick.ac.uk}}}

\date{}

\begin{document}

\maketitle

\begin{abstract}
    The \emph{metric $k$-median} problem is a textbook clustering problem. As input, we are given a metric space $V$ of size $n$ and an integer $k$, and our task is to find a subset $S \subseteq V$ of at most $k$ `centers' that minimizes the total distance from each point in $V$ to its nearest center in $S$.

    Mettu and Plaxton [UAI'02] gave a \textbf{randomized} algorithm for $k$-median that computes a $O(1)$-approximation in $\tilde O(nk)$ time.\footnote{We use $\tilde O(\cdot)$ to hide polylog factors in the size $n$ and the aspect ratio $\Delta$ (see \Cref{sec:prelim}) of the metric space.} They also showed that any algorithm for this problem with a bounded approximation ratio must have a running time of $\Omega(nk)$. Thus, the running time of their algorithm is optimal up to polylogarithmic factors.

    For \textbf{deterministic} $k$-median, Guha et al.~[FOCS'00] gave an algorithm that computes a 
    $\poly(\log (n/k))$-approximation in $\tilde O(nk)$ time, where the degree of the polynomial in the approximation is unspecified. To the best of our knowledge, this remains the state-of-the-art approximation of any deterministic $k$-median algorithm with this running time. 
    
    This leads us to the following natural question: \textbf{What is the best approximation of a deterministic $k$-median algorithm with near-optimal running time?} We make progress in answering this question by giving a deterministic algorithm that computes a $O(\log(n/k))$-approximation in $\tilde O(nk)$ time. We also provide a lower bound showing that any deterministic algorithm with this running time must have an approximation ratio of $\Omega(\log n/(\log k + \log \log n))$, \emph{establishing a gap between the randomized and deterministic settings for $k$-median}.
\end{abstract}

\section{Introduction}\label{sec:intro}

Clustering data is one of the fundamental tasks in unsupervised learning. As input, we are given a dataset, and our task is to partition the elements of the dataset into groups called \emph{clusters} so that similar elements are placed in the same cluster and dissimilar elements are placed in different clusters. One of the basic formulations of clustering is \emph{metric $k$-clustering}, where we are given a (weighted) metric space $(V,w,d)$ of size $n$, and our goal is to find a subset $S \subseteq V$ of at most $k$ \emph{centers} that minimizes an objective function.
We focus on the \emph{$k$-median} problem, where the objective is defined as $\cl(S) := \sum_{x \in V} w(x) \cdot d(x, S)$, where $d(x, S) := \min_{y \in S} d(x,y)$.
Equivalently, we want to minimize the total weighted distance from points in $V$ to their nearest center in $S$.

\medskip
\noindent \textbf{The State-of-the-Art for $k$-Median.} Metric $k$-median is a fundamental clustering problem with many real-world applications and has been studied extensively across many computational models 
\cite{charikar1999constant,jain2001approximation,ahmadian2019better,byrka2017improved,charikar2003better, ailon2009streaming}. The problem is NP-Hard and there is a long line of work designing efficient approximation algorithms for $k$-median using a variety of different techniques, such as local search \cite{AryaGKMMP04} and Lagrangian relaxation \cite{jain2001approximation}.
Mettu and Plaxton gave a randomized algorithm for $k$-median that computes a $O(1)$-approximation in $\tilde O(nk)$ time \cite{MettuP02}, where the approximation guarantee holds with high probability. 
They also showed that any algorithm for this problem with a non-trivial approximation ratio must have a running time of $\Omega (nk)$.
It follows that their algorithm is \emph{near-optimal}, i.e.~optimal up to polylogarithmic factors in the running time and the constant in the approximation ratio.

\medskip
\noindent \textbf{Deterministic Algorithms for $k$-Median.}
Designing \emph{deterministic} algorithms for fundamental problems is an important research direction within algorithms \cite{focs/Cohen-AddadSS23,focs/HaeuplerLS24,stoc/AssadiCS22,talg/NeimanS16, HLRW24}.
Even though the randomized complexity of $k$-median is well understood, we do not have the same understanding of the problem in the deterministic setting.
For deterministic $k$-median,
Mettu and Plaxton gave an algorithm that computes a $O(1)$-approximation in $\tilde O(n^2)$ time \cite{MettuP00}, and Jain and Vazirani gave an algorithm with an improved approximation of $6$ and a running time of $\tilde O(n^2)$ \cite{jain2001approximation}. 
Whenever $k = \Omega(n)$, it follows from the lower bound of \cite{MettuP02} that these algorithms are near-optimal in both the approximation ratio and running time. 
On the other hand, for $k \ll n$, these algorithms are slower than the randomized $O(1)$-approximation algorithm of \cite{MettuP02}.
Guha et al.~gave an algorithm that computes a $\poly(\log(n/k))$-approximation in a near-optimal running time of $\tilde O(nk)$ \cite{focs/GuhaMMO00}, where the degree of the polynomial in the approximation is unspecified. However, it is not clear how much this approximation ratio can be improved, and in particular, whether or not we can match the bounds for randomized algorithms. This leads us to the following question.

\begin{wrapper}
\begin{question}\label{Q1}
\begin{center}
What is the best approximation of any deterministic algorithm for $k$-median that runs in $\tilde O(nk)$ time?
\end{center}
\end{question}
\end{wrapper}

\subsection{Our Results}

We make progress in answering \Cref{Q1} by giving a deterministic algorithm with near-optimal running time and an improved approximation of $O(\log(n/k))$, proving the following theorem.

\begin{theorem}\label{thm:main:fast}
    There is a deterministic algorithm for $k$-median that, given a metric space of size $n$, computes an $O(\log(n/k))$-approximate solution in $\tilde O(nk)$ time.
\end{theorem}

We obtain our algorithm by adapting the ``hierarchical partitioning'' approach of Guha et al.~\cite{focs/GuhaMMO00}. We show that a modified version of this hierarchy can be implemented efficiently by using ``restricted $k$-clustering'' algorithms---a notation that was recently introduced by Bhattacharya et al.~to design fast dynamic clustering algorithms \cite{focs/BCLP24}. We design a deterministic algorithm for restricted $k$-median based on the reverse greedy algorithm of Chrobak et al.~\cite{ChrobakKY06} and combine it with the hierarchical partitioning framework to construct our algorithm.

In addition to our algorithm, we also provide a lower bound on the approximation ratio of any deterministic algorithm with a running time of $\tilde O(nk)$, proving the following theorem.

\begin{theorem}\label{thm:main:lower}
    Any deterministic algorithm for $k$-median that runs in $\tilde O(nk)$ time when given a metric space of size $n$ has an approximation ratio of
    $$ \Omega \! \left( \frac{\log n}{\log k + \log \log n} \right). $$
\end{theorem}

\emph{This lower bound establishes a separation between the randomized and deterministic settings for $k$-median}---ruling out the possibility of a deterministic $O(1)$-approximation algorithm that runs in near-optimal $\tilde O(nk)$ time for $k = n^{o(1)}$. For example, when $k = \poly(\log n)$, 
\Cref{thm:main:lower} shows that any deterministic algorithm with a near-optimal running time must have an approximation ratio of $\Omega(\log n / \log \log n)$. 
On the other hand, \Cref{thm:main:fast} gives such an algorithm with an approximation ratio of $O(\log n)$, \emph{which matches the lower bound up to a lower order $O (\log \log n)$ term}.

We prove \Cref{thm:main:lower} by adapting a lower bound on the \emph{query complexity} of dynamic $k$-center given by Bateni et al.~\cite{BateniEFHJMW23}, where the query complexity of an algorithm is the number of queries that it makes to the distance function $d(\cdot, \cdot)$.
Our lower bound holds for any deterministic algorithm with a query complexity of $\tilde O(nk)$.
Since the query complexity of an algorithm is a lower bound on its running time, this gives us \Cref{thm:main:lower}.
In general, establishing a gap between the deterministic and randomized query complexity of a problem is an interesting research direction \cite{MN20}. Our lower bound implies such a gap for $k$-median when $k$ is sufficiently small.

For the special case of $1$-median, Chang showed that, for any constant $\epsilon > 0$, any deterministic algorithm with a running time of $O(n^{1 + \epsilon})$ has an approximation ratio of $\Omega (1/\epsilon)$ \cite{Chang16}. The lower bound for $1$-median by \cite{Chang16} uses very similar techniques to the lower bounds of Bateni et al.~\cite{BateniEFHJMW23}, which we adapt to obtain our result. In \Cref{thm:lower bound}, we provide a generalization of the lower bound in \Cref{thm:main:lower}, giving a similar tradeoff between running time and approximation.

\medskip
\noindent\textbf{Our Results for $k$-Means.}
Another related clustering problem is \textit{metric $k$-means}, where the objective is defined as $\cl(S) := \sum_{x \in V} w(x) \cdot d(x, S)^2$.
For $k$-means, the current state-of-the-art is essentially the same as for $k$-median.
Using randomization, it is known how to obtain a $O(1)$-approximation in $\tilde O(nk)$ time \cite{MettuP02}.
In \Cref{sec:guha alg}, we describe a generalization of the deterministic algorithm of \cite{focs/GuhaMMO00} and show that it works for $k$-means as well as $k$-median, giving a $\poly(\log(n/k))$-approximation for $k$-means in near-optimal $\tilde O(nk)$ time.

Both our algorithm and lower bound for $k$-median extend to $k$-means as well.
The following theorems summarize our results for deterministic $k$-means.
We describe how to extend our results to $k$-means in \Cref{sec:out-k-means}.

\begin{theorem}
    There is a deterministic algorithm for $k$-means that, given metric space of size $n$, computes an $O(\log^2(n/k))$-approximate solution in $\tilde{O}(nk)$ time.
\end{theorem}

\begin{theorem}
    Any deterministic algorithm for $k$-means that runs in $\tilde{O}(nk)$ time when given a metric space of size $n$ has an approximation ratio of
    $$ \Omega \! \left( \left( \frac{\log n}{\log k + \log\log n} \right)^2 \right). $$
\end{theorem}

\subsection{Related Work}\label{sec:related}

Another well-studied metric $k$-clustering problems related to $k$-median and $k$-means is \emph{$k$-center}.
For $k$-center, the situation is quite different. The classic greedy algorithm given by Gonzalez \cite{tcs/Gonzalez85} is deterministic and returns a $2$-approximation in $O(nk)$ time.
It is known that any non-trivial approximation algorithm must run in $\Omega(nk)$ time \cite{BateniEFHJMW23}, and it is NP-Hard to obtain a $(2 - \epsilon)$-approximation for any constant $\epsilon > 0$ \cite{dam/HsuN79}. Thus, this algorithm has an exactly optimal approximation ratio and running time (assuming $\text{P} \neq \text{NP}$).

Many specific cases and generalizations of $k$-median have also been considered. A particularly important line of work considers the specific case of Euclidean spaces. It was recently shown how to obtain a $\poly(1/\epsilon)$-approximation in $\tilde O(n^{1 + \epsilon + o(1)})$ time in such spaces \cite{DupreS24}. They obtain their result by adapting the $\tilde O(n^2)$ time deterministic algorithm of \cite{MettuP00} using locality-sensitive hashing.
The more general \emph{non-metric} $k$-median problem, where the distances between points do not have to satisfy the triangle inequality, has also been considered. Recently, \cite{young25} designed a $\tilde{O}(n^2k)$ time algorithm for computing a \emph{$O(\log(n/k))$-size-approximation}, where the cost of the returned solution is at most the cost of the optimal solution of size $k$ (i.e.~$\OPT_k$) and the size of the solution is at most $O(\log(n/k)) \cdot k$.

The $k$-median problem has also recently received much attention in the \emph{dynamic} setting, where points in the metric space are inserted and deleted over time and the objective is to maintain a good solution. A long line of work \cite{nips/Cohen-AddadHPSS19,HenzingerK20,ourneurips2023,esa/TourHS24,focs/BCLP24,ourstoc25} recently led to a fully dynamic $k$-median algorithm with $O(1)$-approximation and $\tilde O(k)$ update time against adaptive adversaries, giving near-optimal update time and approximation.\footnote{Note that, since we cannot obtain a running time of $o(nk)$ in the static setting, we cannot obtain an update time of $o(k)$ in the dynamic setting.}

\subsection{Organization}
In \Cref{sec:prelim}, we give the preliminaries and describe the notation used throughout the paper.
In \Cref{sec:tech}, we give a technical overview of our results.
We present our algorithm in \Cref{sec:restr,sec:our alg}.
Our lower bound is described in \Cref{sec:lower}.
Finally, in \Cref{sec:out-k-means}, we describe our results for the $k$-means problem.

\section{Preliminaries}\label{sec:prelim}
Let $(V, w, d)$ be a weighted metric space of size $n$, where $w: V \longrightarrow \mathbb R_{\geq 0}$ is a weight function and $d: V \times V \longrightarrow \mathbb R_{\geq 0}$ is a metric satisfying the triangle inequality. The aspect ratio $\Delta$ of the metric space is the ratio of the maximum and minimum non-zero distances in the metric space. We use the notation $\tilde O(\cdot)$ to hide polylogarithmic factors in the size $n$ and the aspect ratio $\Delta$ of the metric space.
Given subsets $S, U \subseteq V$, we define the cost of the solution $S$ with respect to $U$ as
$$ \cost(S, U) := \sum_{x \in U} w(x) \cdot d(x,S), $$
where $d(x, S) := \min_{y \in S} d(x,y)$.\footnote{Note that we do not necessarily require that $S$ is a subset of $U$.} When we are considering the cost of $S$ w.r.t.~the entire space $V$, we abbreviate $\cost(S,V)$ by $\cost(S)$. 
In the $k$-median problem on the metric space $(V,w,d)$, our objective is to find a subset $S \subseteq V$ of size at most $k$ which minimizes $\cost(S)$.
Given an integer $k \geq 1$ and subsets $X, U \subseteq V$, we define the optimal cost of a solution of size $k$ within $X$ with respect to $U$ as
$$\OPT_k(U, X) := \min_{S \subseteq X, \, |S| = k} \cost(S, U).$$
When $X$ and $U$ are the same, we abbreviate $\OPT_k(U,X)$ by $\OPT_k(U)$.
Thus, the optimal solution to the $k$-median problem on the metric space $(V,w,d)$ has cost $\OPT_k(V)$.
For any $U \subseteq V$, we denote the metric subspace obtained by considering the metric $d$ and weights $w$ restricted to only the points in $U$ by $(U,w,d)$.

\medskip
\noindent \textbf{The Projection Lemma.} Given sets $A, B \subseteq V$, we let $\proj(A,B)$ denote the projection of $A$ onto the set $B$, which is defined as the subset of points $y \in B$ such that some point $x \in A$ has $y$ as its closest point in $B$ (breaking ties arbitrarily). In other words, we define
$ \proj(A,B) := \left\{\arg \min_{y \in B} d(y, x) \; \middle| \; x \in A\right\}$.
We use the following well-known \emph{projection lemma} throughout the paper, which allows us to upper bound the cost of the projection $\proj(A,B)$ in terms of the costs of $A$ and $B$ \cite{GuptaT08, ChrobakKY06}.
\begin{lemma}\label{lem:projection}
    For any subsets $A, B \subseteq V$, we have that $\cost(\proj(A,B)) \leq \cost(B) + 2  \cdot \cost(A)$.
\end{lemma}

\begin{proof}
    Let $C$ denote $\proj(A,B)$. Let $x \in V$ and let $y^\star$ and $y$ be the closest points to $x$ in $A$ and $B$ respectively. Let $y'$ be the closest point to $y^\star$ in $C$. Then we have that
    $$ d(x,C) \leq d(x, y') \leq d(x, y^\star) + d(y^\star, y') \leq d(x, y^\star) + d(y^\star, y) \leq d(x, y) + 2 \cdot d(x, y^\star), $$
    and so $d(x, C) \leq d(x, B) + 2 \cdot d(x, A)$. It follows that 
    $$\cost(C) = \sum_{x \in V} w(x) d(x,C) \leq \sum_{x \in V} w(x) (d(x, B) + 2 \cdot d(x, A)) = \cost(B) + 2  \cdot \cost(A).\qedhere$$
\end{proof}

The following well-known corollary of the projection lemma shows that, for any set $U \subseteq V$, the optimal cost of the $k$-median problem in $(U,w,d)$ changes by at most a factor of $2$ if we are allowed to place centers anywhere in $V$.

\begin{corollary}\label{cor:improper}
    For any subset $U \subseteq V$, we have that $\OPT_k(U) \leq 2 \cdot \OPT_k(U,V)$.
\end{corollary}

\begin{proof}
    Let $S_V^\star$ be a subset of $V$ of size at most $k$ that minimizes $\cost(S_V^\star, U)$ and let $S_U^\star = \proj(S_V^\star, U)$. Then, for any $x \in U$, it follows from \Cref{lem:projection} that $d(x, S^\star_U) \leq d(x,U) + 2\cdot d(x, S^\star_V) = 2\cdot d(x, S^\star_V)$, which implies the corollary.
\end{proof}

\section{Technical Overview}\label{sec:tech}

We begin by describing the hierarchical partitioning approach used by Guha et al.~\cite{focs/GuhaMMO00} to obtain a $\poly(\log(n/k))$-approximation algorithm with near-optimal running time. We then discuss the limitations of this approach and describe how we overcome these limitations to obtain our result.

\subsection{The Hierarchical Partitioning Framework}

Guha et al.~\cite{focs/GuhaMMO00} showed how to combine an $\tilde O(n^2)$ time $k$-median algorithm with a simple hierarchical partitioning procedure in order to produce a faster algorithm---while incurring some loss in the approximation. Their approach is based on the following divide-and-conquer procedure: 
\begin{enumerate}
    \item Partition the metric space $(V,w,d)$ into $q$ metric subspaces $(V_1,w,d), \dots, (V_q,w,d)$.
    \vspace{-0.2cm}
    \item Solve the $k$-median problem on each subspace $(V_i,w,d)$ to obtain a solution $S_i \subseteq V_i$.
    \vspace{-0.2cm}
    \item Combine the solutions $S_1,\dots,S_q$ to get a solution $S$ for the original space $(V,w,d)$.\label{step:merge}
\end{enumerate}
The main challenge in this framework is implementing Step~\ref{step:merge}---finding a good way to merge the solutions from the subspaces into a solution for the original space. To implement this step, they prove the following lemma, which, at a high level, shows how to use the solutions $S_i$ to construct a \emph{sparsifier} for the metric space $(V,w,d)$ that is much smaller than the size of the space.

\begin{lemma}[\cite{focs/GuhaMMO00}]\label{lem:intro:sparse}
    Suppose that each solution $S_i$ is a $\beta$-approximate solution to the $k$-median problem in $(V_i, w,d)$.
    Let $V' = \bigcup_i S_i$ and, for each $y \in S_i$, let $w'(y)$ denote the total weight of points in $V_i$ that are assigned to $y$ in the solution $S_i$. Then any $\alpha$-approximate solution $S$ to the $k$-median problem in the space $(V', w', d)$ is a $O(\alpha\beta)$-approximation in the space $(V,w,d)$.
\end{lemma}

Using \Cref{lem:intro:sparse}, we can compute a (weighted) subspace $(V', w', d)$ that has size only $\sum_i |S_i| = O(kq)$. Crucially, we have the guarantee that any good solution that we find within this subspace is also a good solution in the space $(V,w,d)$. Thus, we can then use the deterministic $O(1)$-approximate $\tilde O(n^2)$ time $k$-median algorithm of $\cite{MettuP00}$ to compute a solution $S$ for $(V', w', d)$ in $\tilde O(k^2 q^2)$ time.

\medskip
\noindent \textbf{A 2-Level Hierarchy.} Suppose we run this divide-and-conquer framework for one step (i.e.~without recursing on the subspaces $(V_i, w, d)$) and just compute the solutions $S_i$ for $(V_i, w, d)$ using the $\tilde O(n^2)$ time algorithm of \cite{MettuP00}. It follows immediately from \Cref{lem:intro:sparse} that the approximation ratio of the final solution $S$ is $O(1)$. We can also observe that, up to polylogarithmic factors, the total time taken to compute the $S_i$ is $\simeq q \cdot (n/q)^2 = n^2/q$, since the size of each subspace is $O(n/q)$. Furthermore, the time taken to compute $S$ is $\tilde O( k^2q^2)$. By taking $q = (n/k)^{2/3}$, we get that the total running time of the algorithm is $\tilde O(nk \cdot (n/k)^{1/3})$, giving a polynomial improvement in the running time for $k \ll n$.

\medskip
\noindent \textbf{An $\ell$-Level Hierarchy.} Now, suppose we  run this framework for $\ell$ steps. To balance the running time required to compute the solutions at each level of this divide-and-conquer procedure, we want to subdivide each metric subspace at depth $i$ in the recursion tree into (roughly) $q_i = (n/k)^{2^{-i}}$ further subspaces. Guha et al.~show that the running time of this algorithm is $\tilde O(nk \cdot (n/k)^{2^{-\ell}})$. By \Cref{lem:intro:sparse}, we can also see that the approximation ratio of the final solution $S$ is $2^{O(\ell)}$. Setting $\delta = (n/k)^{2^{-\ell}}$, we get the following theorem.
\begin{theorem}\label{thm:intro:guha}
    There is a deterministic algorithm for $k$-median that, given a metric space of size $n$, computes a $\poly(\log(n/k) / \log \delta)$-approximation in $\tilde O(nk\delta)$ time, for any $2 \leq \delta \leq n/k$.
\end{theorem}

Setting $\delta = O(1)$, we get immediately the following corollary.

\begin{corollary}
    There is a deterministic algorithm for $k$-median that, given a metric space of size $n$, computes a $\poly(\log(n/k))$-approximation in $\tilde O(nk)$ time.
\end{corollary}

We remark that the results in \cite{focs/GuhaMMO00} are presented differently and only claim an approximation ratio of $\poly(\log n /\log \delta)$. In \Cref{sec:guha alg}, we describe a generalization of their algorithm, proving \Cref{thm:intro:guha} and also showing that it extends to $k$-means.

\subsection{The Barrier to Improving The Approximation}

While the sparsification technique that is described in \Cref{lem:intro:sparse} allows us to obtain much faster algorithms by sparsifying our input in a hierarchical manner, this approach has one major drawback. Namely, the fact that sparsifying the input in this manner also leads to an approximation ratio that scales exponentially with the number of levels in the hierarchy. Unfortunately, this exponential growth seems unavoidable with this approach. This leads us to the following question.
\begin{wrapper}
\begin{question}\label{Q2}
\begin{center}
Is there a different way to combine the solutions $S_1, \dots,S_q$ that does \textbf{\emph{not}} lead to exponential deterioration in the approximation?
\end{center}
\end{question}
\end{wrapper}

\subsection{Idea I: Sparsification via Restricted $k$-Median}

Very recently, Bhattacharya et al.~introduced the notion of ``restricted $k$-clustering'' in order to design efficient and consistent dynamic clustering algorithms \cite{focs/BCLP24}. The \textbf{restricted $k$-median} problem on the space $(V,w,d)$ is the same as the $k$-median problem, except that we are also given a subset $X \subseteq V$ and have the additional restriction that our solution $S$ must be a subset of $X$. Crucially, even though the algorithm can only place centers in the solution $S$ within $X$, it receives the entire space $V$ as input and computes the cost of the solution $S$ w.r.t.~the entire space. 

The restricted $k$-median problem allows us to take a different approach toward implementing Step~\ref{step:merge} of the divide-and-conquer framework. Instead of compressing the entire space $(V,w,d)$ into only $O(kq)$ many weighted points, we restrict the output of the solution $S$ to these $O(kq)$ points but still consider the rest of the space while computing the cost of the set $S$.
This can be seen as a less aggressive way of sparsifying the metric space, where we lose less information.
It turns out that this approach allows us to produce solutions of higher quality, where the approximation scales linearly in the number of levels in the hierarchy instead of exponentially.

Efficiently implementing this new hierarchy is challenging since we need to design a fast deterministic algorithm for restricted $k$-median with the appropriate approximation guarantees.
To illustrate our approach while avoiding this challenge, we first describe a simple version of our algorithm with improved approximation and near-optimal \emph{query complexity}. We later show how to design such a restricted algorithm, allowing us to implement this algorithm efficiently.

\subsection{Our Algorithm With Near-Optimal Query Complexity}\label{sec:low query}

Let $(V, w, d)$ be a metric space of size $n$ and $k \leq n$ be an integer.
We define the value $\ell := \log_2(n/k)$, which we use to describe our algorithm.
Our algorithm works in the following 2 \emph{phases}.

\medskip
\noindent
\textbf{Phase I:} In the first phase of our algorithm, we construct a sequence of partitions $Q_0,\dots, Q_\ell$ of the metric space $V$, such that the partition $Q_i$ is a \emph{refinement} of the partition $Q_{i-1}$.\footnote{i.e.~for each element $X \in Q_{i-1}$, there are elements $X_1,\dots, X_q \in Q_i$ such that $X=X_1\cup \dots \cup X_q$.} We let $Q_0 := \{V\}$. Subsequently, for each $i = 1,\dots,\ell$, we construct the partition $Q_i$ by arbitrarily partitioning each $X \in Q_{i-1}$ into subsets $X_1$ and $X_{2}$ of equal size and adding these subsets to $Q_i$.\footnote{Note that it might not be possible for the subsets $X_1$ and $X_{2}$ to have equal size.
For simplicity, we ignore this detail in the technical overview.}
For $X \in Q_{i-1}$, we define $\mathcal P(X) := \{X' \in Q_{i} \mid X' \subseteq X\}$.

\medskip
\noindent
\textbf{Phase II:}
The second phase of our algorithm proceeds in \emph{iterations}, where we use the partitions $\{Q_i\}_i$ to compute the solution in a bottom-up manner.
Let $V_{\ell + 1}$ denote the set of points $V$.
For each $i = \ell ,\dots, 0$, our algorithm constructs $V_i$ as follows:

\begin{wrapper}
    For each $X \in Q_i$, let $S_X$ be the optimal solution to the $k$-median problem in the metric space $(X,w,d)$ such that $S_X \subseteq X \cap V_{i+1}$. Let $V_i := \bigcup_{X \in Q_i} S_X$.
\end{wrapper}

\medskip
\noindent
\textbf{Output:}
The set $V_0$ is the final output of our algorithm.

\medskip

The following theorem summarizes the behaviour of this algorithm.

\begin{theorem}\label{thm:intro:query}
    There is a deterministic algorithm for $k$-median that, given a metric space of size $n$, computes a $O(\log(n/k))$-approximation with $\tilde O(nk)$ queries (but exponential running time).
\end{theorem}

We now sketch the proof of \Cref{thm:intro:query} by outlining the analysis of the approximation ratio and query complexity. The formal proof follows from the more general analysis in \Cref{sec:anal:ouralg}.

\subsubsection*{Approximation Ratio}

To bound the cost of the solution $V_0$ returned by our algorithm, we first need to be able to relate the costs of a solution in the hierarchy to the costs of the solutions in the subsequent levels.
Given any set $X$ within a partition $Q_i$, the following key claim establishes the relationship between the cost of the solution $S_X$ on the metric subspace $(X,w,d)$ and the costs of the solutions $\{S_{X'}\}_{X' \in \mathcal P(X)}$ on the metric subspaces $\{(X',w,d)\}_{X' \in \mathcal P(X)}$.

\begin{claim}\label{claim:apx:1}
For any set $X \in \bigcup_{i=0}^{\ell - 1} Q_i$, we have that
$$ \cost(S_X, X) \leq \sum_{X' \in \mathcal P(X)} \cost(S_{X'}, X') + O(1) \cdot \OPT_k(X). $$
\end{claim}

\begin{proof}
    Let $R$ denote the set $\bigcup_{X' \in \mathcal P(X)} S_{X'}$. Let $S^\star$ be an optimal solution to the $k$-median problem in the metric space $(X,w,d)$ and let $S' := \proj(S^\star, R)$ be the projection of $S^\star$ onto $R$. 
    
    Since $S_X$ is the optimal solution to the $k$-median problem in $(X,w,d)$ such that $S_X \subseteq R$, it follows that $\cost(S_X, X) \leq \cost(S', X)$.
    Applying the projection lemma (\Cref{lem:projection}) to the projection $S'$ of $S^\star$ onto $R$, we get that $\cost(S', X) \leq \cost(R, X) + 2 \cdot \cost(S^\star, X) = \cost(R, X) + O(1) \cdot \OPT_k(X)$.
    Combining these two inequalities, it follows that
    \begin{equation}\label{eq:kay apx}
        \cost(S_X, X) \leq \cost(R, X) + O(1) \cdot \OPT_k(X).
    \end{equation}
    The claim then follows from the fact that 
    $$ \cost(R,X) = \sum_{X' \in \mathcal P(X)} \cost(R, X') \leq \sum_{X' \in \mathcal P(X)} \cost(S_{X'}, X').\qedhere $$
\end{proof}

By repeatedly applying \Cref{claim:apx:1} to the sum $\sum_{X\in Q_i}\cost(S_X, X)$, we obtain the following upper bound on $\cost(V_0)$:
\begin{equation}\label{eq:myeqq}
    \cost(V_0) \leq \sum_{X \in Q_\ell} \cost(S_{X}, X) +  O(\ell) \cdot \OPT_k(V).
\end{equation}
We prove \Cref{eq:myeqq} in \Cref{lem:apx:2}.
Now, let $S^\star$ be an optimal solution to $k$-median in $(V,w,d)$ and consider any $X \in Q_\ell$. Since $S_X$ is an optimal solution to $k$-median in the subspace $(X, w, d)$, we have that $\cost(S_X, X) = \OPT_k(X) \leq 2 \cdot \OPT_k(X, V) \leq 2 \cdot \cost(S^\star, X)$, where the first inequality follows from \Cref{cor:improper}. Thus, summing over each $X \in Q_\ell$, we get that
\begin{equation}\label{eq:myeqq2}
    \sum_{X \in Q_\ell} \cost(S_X, X) \leq 2 \cdot \sum_{X \in Q_\ell} \cost(S^\star, X) = 2 \cdot \cost(S^\star) = 2\cdot \OPT_k(V).
\end{equation}
Combining \Cref{eq:myeqq,eq:myeqq2}, it follows that $\cost(V_0) \leq O(\ell) \cdot \OPT_k(V)$. Thus, the solution $V_0$ returned by our algorithm is a $O(\ell) = O(\log(n/k))$ approximation.

\subsubsection*{Query Complexity}
Since the partitions constructed in Phase I of the algorithm are arbitrary, we do not make any queries to the metric in Phase I.
Thus, we now focus on bounding the query complexity of Phase~II.

The total number of queries made during the $i^{th}$ iteration in Phase II is the sum of the number of queries required to compute the solutions $\{S_X\}_{X \in Q_i}$.
In the first iteration (when $i = \ell$), we compute $|Q_\ell|$ many solutions, each one on a subspace of size $n/|Q_\ell|$.\footnote{Again, we assume for simplicity that each set in the partition has the same size.} We can trivially compute an optimal solution in $(X,w,d)$ by querying every distance between all $O(|X|^2)$ pairs of points in $X$ and then checking every possible solution. Thus, we can upper bound the number of queries by
\begin{equation*}
\left(\frac{n}{|Q_\ell|}\right)^2 \cdot |Q_\ell| 
= \frac{n^2}{|Q_\ell|} 
= \frac{n^2}{2^\ell} = O(nk),
\end{equation*}
where we are using the fact that the number of sets in the partition $Q_\ell$ is $2^\ell = n/k$.
For each subsequent iteration (when $0 \leq i < \ell$), we compute $|Q_i|$ many solutions, each one on a subspace $(X,w,d)$ of size at most $n/|Q_i|$, where the solution is restricted to the set  $X \cap V_{i+1}$, which has size at most
$ |X \cap V_{i+1}| = |S_{X_1} \cup S_{X_2}| \leq 2k, $
where $\calP(X) = \{X_1, X_2\}$ and $S_{X_1}$ and $S_{X_2}$ are computed in the previous iteration.
Since we only need to consider solutions that are contained in some subset of at most $O(k)$ points, we can find an optimal such restricted solution with $O(k) \cdot (n/|Q_i|)$ queries.
It follows that the total number of queries that we make during this iteration is at most
$O(nk/|Q_i|) \cdot |Q_i| \leq O(nk)$.
Thus, the total query complexity of our algorithm is $\ell \cdot O(nk) = \tilde O(nk)$.

\subsection{Idea II: A Deterministic Algorithm for Restricted $k$-Median}

In order to establish the \emph{approximation guarantees} of our algorithm in \Cref{sec:low query}, we use the fact that, given a metric subspace $(V,w,d)$ of size $n$ and a subset $X \subseteq V$, we can find a solution $S \subseteq X$ to the $k$-median problem such that
\begin{equation}\label{eq:watwe need}
    \cost(S) \leq \cost(X) + O(1) \cdot \OPT_k(V).
\end{equation}
We use this fact in the proof of \Cref{claim:apx:1} (see \Cref{eq:kay apx}), which is the key claim in our analysis of our algorithm.
Furthermore, to establish the \emph{query complexity} bound of our algorithm, we use the fact that we can find such a solution $S_X$ with $O(n |X|)$ many queries. To implement this algorithm efficiently, we need to be able to find such a solution in $\tilde O(n |X|)$ \emph{time}. While designing an algorithm with these exact guarantees seems challenging (since it would require efficiently matching the bounds implied by the projection lemma), we can give an algorithm with the following relaxed guarantees, which suffice for our applications.

\begin{lemma}\label{thm:rest:intro}
    There is a deterministic algorithm that, given a metric space $(V,w,d)$ of size $n$, a subset $X \subseteq V$, and a parameter $k$, returns a solution $S \subseteq X$ of size $2k$ in $\tilde O(n |X|)$ time such that
    $$ \cost(S) \leq \cost(X) + O \!\left( \log \! \left(\frac{|X|}{k} \right)\right) \cdot \OPT_k(V). $$
\end{lemma}

To prove \Cref{thm:rest:intro}, we give a simple modification of the reverse greedy algorithm of Chrobak, Kenyon, and Young \cite{ChrobakKY06}. The reverse greedy algorithm initially creates a solution $S \leftarrow V$ consisting of the entire space, and proceeds to repeatedly peel off the point in $S$ whose removal leads to the smallest increase in the cost of $S$ until only $k$ points remain. Chrobak et al.~showed that this algorithm has an approximation ratio of $O(\log n)$ and $\Omega(\log n / \log \log n)$.
While this large approximation ratio might seem impractical for our purposes, we can make 2 simple modifications to the algorithm and analysis in order to obtain the guarantees in \Cref{thm:rest:intro}.
\begin{enumerate}
    \item We start off by setting $S \leftarrow X$ instead of $S \leftarrow V$. This ensures that the output $S$ is a subset of $X$, and gives us the guarantee that $\cost(S) \leq \cost(X) + O(\log |X|) \cdot \OPT_k(V)$.
    \item We return the set $S$ once $|S| \leq 2k$ instead of $|S| \leq k$. This allows us to obtain the guarantee that $\cost(S) \leq \cost(X) + O(\log (|X|/k)) \cdot \OPT_k(V)$.
\end{enumerate}
We provide the formal proof of \Cref{thm:rest:intro} in \Cref{sec:restr}.

We can now make the following key observation: In our algorithm from \Cref{sec:low query}, whenever we compute a solution $S_X$ to the $k$-median problem, we impose the constraint that $S_X \subseteq R$ for a set $R$ of size $R = O(k)$. Thus, if we use the algorithm from \Cref{thm:rest:intro} to compute our solutions within the hierarchy, whenever we apply this lemma we can assume that $|X| = O(k)$. Consequently, the approximation guarantee that we get from \Cref{thm:rest:intro} becomes
$$ \cost(S) \leq \cost(X) + O(1) \cdot \OPT_k(V), $$
matching the required guarantee from \Cref{eq:watwe need}.

\medskip
\noindent \textbf{Dealing with Bicriteria Approximations.} One caveat from \Cref{thm:rest:intro} is that the solutions output by the algorithm have size $2k$ instead of $k$. In other words, these are ``bicriteria approximations'' to the $k$-median problem, i.e.~solutions that contain more than $k$ points. Thus, the final output of our algorithm has size $2k$. By using the extraction technique of Guha et al.~\cite{focs/GuhaMMO00} described in \Cref{lem:intro:sparse}, it is straightforward to compute a subset of $k$ of these points in $\tilde O(k^2)$ time while only incurring a $O(1)$ increase in the approximation ratio.

\medskip
\noindent \textbf{Putting Everything Together.} By combining our algorithm from \Cref{sec:low query} with \Cref{thm:rest:intro}, we get our main result, which we prove in \Cref{sec:our alg}.

\begin{theorem}\label{thm:main:fast:restattet}
    There is a deterministic algorithm for $k$-median that, given a metric space of size $n$, computes a $O(\log(n/k))$-approximate solution in $\tilde O(nk)$ time.
\end{theorem}

\subsection{Our Lower Bound for Deterministic $k$-Median}

We also prove the following lower bound for deterministic $k$-median. Due to space constraints, we defer the proof of this theorem to \Cref{sec:lower}.

\begin{theorem}
    For every $\delta \geq 1$, any deterministic algorithm for the $k$-median problem that has a running time of $O(kn\delta)$ on a metric space of size $n$ has an approximation ratio of
    $$ \Omega \! \left( \frac{\log n}{\log\log n + \log k + \log \delta} \right). $$
\end{theorem}

We prove this lower bound by adapting a lower bound construction by Bateni et al.~\cite{BateniEFHJMW23}.
In this paper, the authors provide a lower bound for dynamic $k$-center clustering against adaptive adversaries.
Although their primary focus is $k$-center, their lower bounds can be extended to various $k$-clustering problems.
The main idea is to design an adaptive adversary that controls the underlying metric space as well as the points being inserted and deleted from the metric space.
Whenever the algorithm queries the distance between any two points $x,y$, the adversary returns a value $d(x,y)$ that is consistent with the distances reported for all previously queried pairs of points.
Note that if we have the distances between some points (not necessarily all of them), we might be able to embed different metrics on the current space that are consistent with the queried distances.
More specifically, \cite{BateniEFHJMW23} introduces two different consistent metrics, and shows that the algorithm can not distinguish between these metrics, leading to a large approximation ratio.

We use the same technique as described above with slight modifications.
The first difference is that \cite{BateniEFHJMW23} considers the problem in the fully dynamic setting, whereas our focus is on the static setting.
The adversary has two advantages in the fully dynamic setting:
\begin{enumerate}
    \item The adversary has the option to delete a (problematic) point from the space.
    \vspace{-0.2cm}
    \item The approximation ratio of the algorithm is defined as the maximum approximation ratio throughout the entire stream of updates.
\end{enumerate}
Both of these advantages are exploited in the framework of \cite{BateniEFHJMW23}:
The adversary removes any `problematic' point $x$ and
the approximation ratio of the algorithm is proven to be large only in special steps (referred to as `clean operations') where the adversary has removed all of the problematic points.
In the static setting, the adversary does not have these advantages, and the approximation ratio of the algorithm is only considered at the very end.

Due to the differences between the objective functions of the $k$-median and $k$-center problems, we observed that we can adapt the above framework for deterministic static $k$-median.
One of the technical points here is to show that, if we have a problematic point $x$ that is contained in the output $S$ of the algorithm, we can construct the metric so that the cost of the cluster of $x$ in solution $S$ become large (see \Cref{lem:number-of-neighbours-of-distance-one}). 
The final metric space is similar to the `uniform' metric introduced in \cite{BateniEFHJMW23} with a small modification.
Since the algorithm is \textbf{deterministic}, the output is the same set $S$ if we run the algorithm on this final metric again.
Hence, we get our lower bounds for any deterministic algorithm for the static $k$-median problem.

\section{A Deterministic Algorithm for Restricted $k$-Median}\label{sec:restr}

In this section, we prove the following theorem.

\begin{theorem}\label{thm:restr}
    There is a deterministic algorithm that, given a metric space $(V,w,d)$ of size $n$, a subset $X \subseteq V$, and a parameter $k' \geq k$, returns a solution $S \subseteq X$ of size $k'$ in $\tilde O(n |X|)$ time such that
    $$ \cost(S) \leq \cost(X) + O \!\left( \log \! \left(\frac{|X|}{k' - k + 1} \right)\right) \cdot \OPT_k(V). $$
\end{theorem}

This algorithm is based on a variant of the reverse greedy algorithm for the $k$-median problem \cite{ChrobakKY06}, which we modify for the restricted $k$-median problem. We refer to our modified version of this algorithm as $\Restr_{k'}$.

\subsection{The Restricted Reverse Greedy Algorithm}

Let $(V, w, d)$ be a metric space of size $n$, $X \subseteq V$ and $k' \leq |X|$ be an integer. The restricted reverse greedy algorithm $\Restr_{k'}$ begins by initializing a set $S \leftarrow X$ and does the following:
\begin{wrapper}
While $|S| > k'$, identify the center $y \in S$ whose removal from $S$ causes the smallest increase in the cost of the solution $S$ (i.e.~the point $y = \arg \min_{z \in S} \cost(S - z)$) and remove $y$ from $S$. Once $|S| \leq k'$, return the set $S$.
\end{wrapper}

\subsection{Analysis}

Let $m$ denote the size of $X$ and let $k \leq k'$.
Now, suppose that we run $\Restr_{k'}$ and consider the state of the set $S$ throughout the run of the algorithm. Since points are only removed from $S$,
this gives us a sequence of nested subsets $S_{k'} \subseteq \dots \subseteq S_{m}$, where $|S_i| = i$ for each $i \in [k', m]$. Note that $S_{k'}$ is the final output of our algorithm. The following lemma is the main technical lemma in the analysis of this algorithm.

\begin{lemma}\label{lem:greedy:main}
    For each $i \in [k' + 1, m]$, we have that $$\cost(S_{i-1}) - \cost(S_i) \leq \frac{2}{i - k} \cdot \OPT_k(V).$$
\end{lemma}

Using \Cref{lem:greedy:main}, we can now prove the desired approximation guarantee of our algorithm.
By using a telescoping sum, we can express the cost of the solution $S_{k'}$ as
\begin{equation}\label{eq:greedy:1}
    \cost(S_{k'}) = \cost(S_{m}) + \sum_{i = k' + 1}^m \left(\cost(S_{i-1}) - \cost(S_i) \right).
\end{equation}
Applying \Cref{lem:greedy:main}, we can upper bound the sum on the RHS of \Cref{eq:greedy:1} by
\begin{equation}\label{eq:greedy:2}
\sum_{i = k' + 1}^m \left(\cost(S_{i-1}) - \cost(S_i) \right) \leq \sum_{i = k' + 1}^m \frac{2}{i - k} \cdot \OPT_k(V) \leq O \! \left(\log \! \left(\frac{m}{k' - k + 1} \right) \right) \cdot \OPT_k(V).
\end{equation}
Combining \Cref{eq:greedy:1,eq:greedy:2}, we get that
$$ \cost(S_{k'}) \leq \cost(S_m) + O \! \left(\log \! \left(\frac{m}{k' - k + 1} \right) \right) \cdot \OPT_k(V), $$
giving us the desired bound in \Cref{thm:restr}, since $S_m = X$ and $m = |X|$.

\subsection{Proof of \Cref{lem:greedy:main}}

The following analysis is the same as the analysis given in \cite{ChrobakKY06}, with some minor changes to accommodate the fact that we are using the algorithm for the restricted $k$-median problem and to allow us to obtain an improved approximation for $k' \geq (1 + \Omega(1))k$, at the cost of outputting bicriteria approximations.

We begin with the following claim, which we use later on in the analysis.

\begin{claim}\label{claim:greedy:proof}
    For all subsets $A \subseteq B \subseteq V$, we have that
    $$ \sum_{y \in B \setminus A} \left( \cost(B-y) - \cost(B) \right) \leq \cost(A) - \cost(B). $$
\end{claim}

\begin{proof}
    For each $y \in B$, let $C_{B}(y)$ denote the cluster of the points in $V$ that are assigned to $y$ within the solution $B$. In other words, $C_B(y)$ is the subset of points in $V$ whose closest point in $B$ is $y$ (breaking ties arbitrarily). We can now observe that
    \begin{align*}
        \sum_{y \in B \setminus A} \left( \cost(B-y) - \cost(B) \right) &= \sum_{y \in B \setminus A} \sum_{x \in V} w(x)\left( d(x, B-y) - d(x,B) \right)\\
        &= \sum_{y \in B \setminus A} \sum_{x \in C_B(y)} w(x) \left( d(x, B-y) - d(x,B) \right)\\
        &\leq \sum_{y \in B \setminus A} \sum_{x \in C_B(y)} w(x) \left( d(x, A) - d(x,B) \right)\\
        &\leq \sum_{x \in V} w(x) \left( d(x, A) - d(x,B) \right)\\
        &= \cost(A) - \cost(B).
    \end{align*}
    Here, the first and fifth lines follow from the definition of the cost function, the second line follows since $d(x, B - y) - d(x, B)$ can only be non-zero when $x \in C_B(y)$, the third line follows from the fact that $A \subseteq B - y$ for any $y \in B \setminus A$, and the fourth line from the fact that $\{C_B(y)\}_{y \in B}$ partition the set $V$.
\end{proof}

Let $S^\star$ denote an optimal solution to the $k$-median problem in the metric space $(V,w,d)$ and let $i \in [k' + 1, m]$. We denote by $S'_i$ the projection $\proj(S^\star, S_i)$ of the optimal solution $S^\star$ onto the set $S_i$. It follows that
\begin{align*}
    \cost(S_{i-1}) - \cost(S_i) &\leq \min_{y \in S_i \setminus S'_i} \left(\cost(S_i - y) - \cost(S_i) \right)\\
    &\leq \frac{1}{|S_i \setminus S'_i|} \cdot \sum_{y \in S_i \setminus S'_i} \left(\cost(S_i - y) - \cost(S_i) \right)\\
    &\leq \frac{1}{i - k} \cdot \sum_{y \in S_i \setminus S'_i} \left(\cost(S_i - y) - \cost(S_i) \right)\\
    &\leq \frac{1}{i - k} \cdot \left(\cost(S'_i) - \cost(S_i) \right)\\
    &\leq \frac{2}{i - k} \cdot \cost(S^\star) = \frac{2}{i - k} \cdot \OPT_k(V).
\end{align*}
The first line follows directly from how the algorithm chooses which point to remove from $S_i$.\footnote{Note that, for analytic purposes, we only take the minimum over $y \in S_i \setminus S_i'$ instead of all of $S_i$ in this inequality.} The second line follows from the fact that the minimum value within a set of real numbers is upper bounded by its average. The third line follows from the fact that $|S_i \setminus S'_i| \geq |S_i| - |S'_i| \geq i - k$. 
The fourth line follows from \Cref{claim:greedy:proof}.
Finally, the fifth line follows from \Cref{lem:projection}, which implies that $\cost(S'_i) \leq \cost(S_i) + 2\cdot \cost(S^\star)$.

\subsection{Implementation}

We now show how to implement $\Restr$ to run in time $O(n |X| \log n)$. 
Our implementation uses similar data structures to the randomized local search in \cite{focs/BCLP24}. 
For each $x \in V$, we initialize a list $L_x$ which contains all of the points $y \in X$, sorted in increasing order according to the distances $d(x,y)$. We denote the $i^{th}$ point in the list $L_x$ by $L_x(i)$.
We maintain the invariant that, at each point in time, each of the lists in $\mathcal L = \{L_x\}_{x \in V}$ contains exactly the points in $S$.
Thus, at each point in time, we have that
$\cost(S) = \sum_{x \in V} w(x)d(x, L_x(1))$.
By implementing each of these lists using a balanced binary tree, we can initialize them in $O(n |X| \log n)$ time and update them in $O(n \log n)$ time after each removal of a point from $S$.
Since the $S$ initially has size $|X|$, the total time spent updating the lists is $O(n|X|\log n)$.
We also explicitly maintain the clustering $\mathcal C = \{C_S(y)\}_{y \in S}$ induced by the lists $\mathcal L$, where $C_S(y) := \{x \in V \mid L_x(1) = y\}$. We can initialize these clusters in $O(n)$ time given the collection of lists $\mathcal L$ and update them each time a list in $\mathcal L$ is updated while only incurring a $O(1)$ factor overhead in the running time. We now show that, using these data structures, we can implement each iteration of the greedy algorithm in $O(n)$ time. Since the algorithm runs for at most $|X|$ iterations, this gives us the desired running time.

\medskip
\noindent \textbf{Implementing an Iteration of Greedy.}
Using the lists $\mathcal L$ and clustering $\mathcal C$, we can compute
$$ \texttt{change}(y) \leftarrow \sum_{x \in C_S(y)} w(x)(d(x,L_x(2)) - d(x,L_x(1)) ) $$
for each $y \in S$.
Since any point $x \in V$ appears in exactly one cluster in $\mathcal C$, this takes $O(n)$ time in total.
By observing that removing $y$ from $S$ causes each point $x \in C_S(y)$ to be reassigned to the center $L_x(2)$,
we can see that $\texttt{change}(y)$ is precisely the value of $\cost(S - y) - \cost(S)$. Since minimizing $\cost(S - y) - \cost(S)$ is equivalent to minimizing $\cost(S - y)$, it follows that
$$ \min_{z \in S}\texttt{change}(z) = \min_{z \in S} (\cost(S - z) - \cost(S)) = \min_{z \in S} \cost(S - z). $$
Thus, we let $y \leftarrow \arg \min_{z \in S}\texttt{change}(z)$, remove $y$ from $S$, and proceed to update the data structures. Excluding the time taken to update the data structures, the iteration takes $O(n)$ time.

\section{Our Deterministic $k$-Median Algorithm}\label{sec:our alg}

In this section, we prove \Cref{thm:main:fast}, which we restate below.

\begin{theorem}\label{thm:main:fast:restate}
    There is a deterministic algorithm for $k$-median that, given a metric space of size $n$, computes a $O(\log(n/k))$-approximate solution in $\tilde O(nk)$ time.
\end{theorem}

\subsection{Our Algorithm}

Let $(V, w, d)$ be a metric space of size $n$ and $k \leq n$ be an integer.
We also define the value $\ell := \lceil \log_2(n/k) \rceil$, which we use to describe our algorithm.
Our algorithm works in 3 \emph{phases}, which we describe below.

\medskip
\noindent
\textbf{Phase I:} In the first phase of our algorithm, we construct a sequence of partitions $Q_0,\dots, Q_\ell$ of the metric space $V$, such that the partition $Q_i$ is a \emph{refinement} of the partition $Q_{i-1}$.\footnote{i.e.~for each element $X \in Q_{i-1}$, there are elements $X_1,\dots, X_q \in Q_i$ such that $X=X_1\cup \dots \cup X_q$.} We start off by setting $Q_0 := \{V\}$. Subsequently, for each $i = 1,\dots,\ell$, we construct the partition $Q_i$ as follows:

\begin{wrapper}
    Initialize $Q_i \leftarrow \varnothing$. Then, for each $X \in Q_{i-1}$, arbitrarily partition $X$ into subsets $X_1$ and $X_{2}$ such that $\left||X_1| - |X_2|\right| \leq 1$, and add these subsets to $Q_i$.
\end{wrapper}
For $X \in Q_{i-1}$, we define $\mathcal P(X) := \{X' \in Q_{i} \mid X' \subseteq X\}$.

\medskip
\noindent
\textbf{Phase II:}
The second phase of our algorithm proceeds in \emph{iterations}, where we use the partitions $\{Q_i\}_i$ to compute the solution in a bottom-up manner.
Let $V_{\ell + 1}$ denote the set of points $V$.
For each $i = \ell ,\dots, 0$, our algorithm constructs $V_i$ as follows:

\begin{wrapper}
    For each $X \in Q_i$, let $S_X$ be the solution obtained by running $\Restr_{2k}$ on the subspace $(X,w,d)$, restricting the output to be a subset of $X \cap V_{i+1}$. Finally, we define $V_i := \bigcup_{X \in Q_i} S_X$.
\end{wrapper}

\medskip
\noindent
\textbf{Phase III:} Consider the set $V_0$
which contains $2k$ points and let $\sigma : V \longrightarrow V_0$ be the projection from $V$ to $V_0$. Define a weight function $w_0$ on each $y \in V_0$ by $w_0(y) := \sum_{x \in \sigma^{-1}(y)} w(x)$ (i.e.~$w_0(y)$ is the total weight of all points in $V$ that are projected onto $y$). Let $S$ be the solution obtained by running the algorithm of Mettu-Plaxton \cite{MettuP00} on the metric space $(V_0, w_0, d)$.

\medskip
\noindent
\textbf{Output:}
The solution $S$ is the final output of our algorithm.

\subsection{Analysis}\label{sec:anal:ouralg}

We now analyze our algorithm by bounding its approximation ratio and running time.

\subsubsection*{Approximation Ratio}

We begin by proving the following claim, which, for any set $X$ within a partition $Q_i$, allows us to express the cost of the solution $S_X$ w.r.t.~the metric subspace $(X,w,d)$ in terms of the costs of the solutions $\{S_{X'}\}_{X' \in \mathcal P(X)}$. 

\begin{claim}\label{lem:apx:1}
For any set $X \in \bigcup_{i=0}^{\ell - 1} Q_i$, we have that
$$ \cost(S_X, X) \leq \sum_{X' \in \mathcal P(X)} \cost(S_{X'}, X') + O(1) \cdot \OPT_k(X). $$
\end{claim}

\begin{proof}
    Let $R$ denote the set $\bigcup_{X' \in \mathcal P(X)} S_{X'}$. We obtain the solution $S_X$ by calling $\Restr_{2k}$ on the metric space $(X,w,d)$ while restricting the output to be a subset of  $R$. Thus, it follows from \Cref{thm:restr} that
    \begin{equation*}\label{eq:apx:1}
        \cost(S_X, X) \leq \cost(R, X) + O \!\left( \log \! \left(\frac{|R|}{2k - k + 1} \right)\right) \cdot \OPT_k(X).
    \end{equation*}
    By observing that $|R|/(k+1) \leq 4k / (k + 1) \leq 4$,
    it follows that $\cost(S_X, X) \leq \cost(R, X) + O(1) \cdot \OPT_k(X)$.
    Finally, the claim follows since
    $$ \cost(R,X) = \sum_{X' \in \mathcal P(X)} \cost(R, X') \leq \sum_{X' \in \mathcal P(X)} \cost(S_{X'}, X').\qedhere $$
\end{proof}

Using \Cref{lem:apx:1}, we now prove the following claim.

\begin{claim}\label{lem:apx:2}
For any $i \in [0, \ell]$, we have that
$$ \cost(V_0, V) \leq \sum_{X \in Q_i} \cost(S_{X}, X) + O(i) \cdot \OPT_k(V). $$
\end{claim}

\begin{proof}
    We prove this claim by induction. Note that the base case where $i = 0$ holds trivially. Now, suppose that the claim holds for some $i - 1 \in [0, \ell - 1]$. 
    Then we have that
    \begin{align*}
        \cost(V_0, V) &\leq \sum_{X \in Q_{i-1}} \cost(S_{X}, X) + O(i-1) \cdot \OPT_k(V)\\
        &\leq \sum_{X \in Q_{i-1}} \left( \sum_{X' \in \mathcal P(X)} \cost(S_{X'}, X') + O(1) \cdot \OPT_k(X) \right) + O(i-1) \cdot \OPT_k(V)\\
        &= \sum_{X \in Q_{i-1}}  \sum_{X' \in \mathcal P(X)} \cost(S_{X'}, X') + O(1) \cdot \sum_{X \in Q_{i-1}} \OPT_k(X) + O(i-1) \cdot \OPT_k(V)\\
        &\leq \sum_{X \in Q_{i}} \cost(S_{X}, X) + O(1) \cdot \OPT_k(V) + O(i-1) \cdot \OPT_k(V)\\
        &= \sum_{X \in Q_{i}} \cost(S_{X}, X) + O(i) \cdot \OPT_k(V).
    \end{align*}
    The second line follows from \Cref{lem:apx:1}.
    In the fourth line, we are using the fact that, for an optimal solution $S^\star$ of $V$ and any partition $Q$ of $V$, we have
    $$\sum_{X \in Q} \OPT_k(X) \leq 2 \cdot \sum_{X \in Q} \OPT_k(X,V) \leq 2\cdot \sum_{X \in Q} \cost(S^\star, X) = 2 \cdot \OPT_k(V),$$
    where the first inequality follows from \Cref{cor:improper}.
\end{proof}

We get the following immediate corollary from \Cref{lem:apx:2} by setting $i = \ell$.

\begin{corollary}\label{eq:aa}
We have that
$$\cost(V_0, V) \leq \sum_{X \in Q_\ell} \cost(S_{X}, X) + O \!\left( \log \! \left(\frac{n}{k} \right)\right) \cdot \OPT_k(V).$$
\end{corollary}

Using \Cref{eq:aa}, we prove the following lemma.

\begin{lemma}\label{lem:bound on V0}
    We have that $\cost(V_0) = O(\log(n/k)) \cdot \OPT_k(V)$.
\end{lemma}

\begin{proof}
    By \Cref{thm:restr}, it follows that, for any $X \in Q_\ell$,
\begin{equation}\label{eq:aaa}
    \cost(S_X, X) \leq \cost(X,X) + O \!\left( \log \! \left(\frac{|X|}{k} \right)\right) \cdot \OPT_k(X) \leq O \!\left( \log \! \left(\frac{n}{k} \right)\right) \cdot \OPT_k(X,V).
\end{equation}
Combining \Cref{eq:aa} and \Cref{eq:aaa}, we get that
$$ \cost(V_0, V) \leq \sum_{X \in Q_\ell} \cost(S_{X}, X) + O \!\left( \log \! \left(\frac{n}{k} \right)\right) \cdot \OPT_k(V) $$
$$ \leq O \!\left( \log \! \left(\frac{n}{k} \right)\right) \cdot \sum_{X \in Q_\ell} \OPT_k(X,V) + O \!\left( \log \! \left(\frac{n}{k} \right)\right) \cdot \OPT_k(V) \leq  O \!\left( \log \! \left(\frac{n}{k} \right)\right) \cdot \OPT_k(V). \qedhere $$
\end{proof}

By \Cref{lem:bound on V0}, we get that $V_0$ is a $O(\log(n/k))$-bicriteria approximation of size $2k$. Using the extraction technique of \cite{focs/GuhaMMO00} (see \Cref{lem:intro:sparse} or \Cref{lem:bicri=sparsifier}), which allows us to compute an exact solution to the $k$-median problem from a bicriteria approximation while only incurring constant loss in the approximation ratio, it follows that the solution $S$ constructed in Phase~III is a $O(\log(n/k))$-approximation and has size at most $k$.

\subsubsection*{Running Time}

We begin by proving the following lemma, which summarizes the relevant properties of the partitions constructed in Phase I of the algorithm.

\begin{lemma}\label{lem:partitions:main}
    For each $i \in [0, \ell]$, the set $Q_i$ is a partition of $V$ into $2^i$ many subsets of size at most $n/|Q_i| + 2$.
\end{lemma}

\begin{proof}
    We define $Q_0$ as $\{V\}$, which is a trivial partition of $V$. Now, suppose that this statement holds for the partition $Q_i$, where $0 \leq i < \ell$. The algorithm constructs $Q_{i+1}$ by taking each $X \in Q_i$ and further partitioning $X$ into subsets $X_1$ and $X_{2}$, such that difference in the sizes of these subsets is at most $1$.
    We can also observe that the number of subsets in the partition $Q_{i+1}$ is $2 \cdot |Q_{i}| = 2 \cdot 2^{i} = 2^{i+1}$.\footnote{Note that we do not necessarily guarantee that all of the sets in these partitions are non-empty.} Since each subset $X \in Q_i$ has size at most $n/|Q_i| + 2$, it follows that each subset in $Q_{i+1}$ has size at most
    $$ \left\lceil \frac{1}{2} \cdot \left( \frac{n}{|Q_i|} + 2 \right) \right\rceil \leq \frac{n}{2 \cdot |Q_i|} + \frac{2}{2} + 1 \leq \frac{n}{|Q_{i+1}|} + 2.\qedhere $$
\end{proof}

\medskip
\noindent \textbf{Bounding the Running Time.}
We now bound the running time of our algorithm.
The running time of Phase I of our algorithm is $O(n \ell) = \tilde O(n)$, since it takes $O(n)$ time to construct each partition $Q_i$ given the partition $Q_{i-1}$. The running time of Phase III of our algorithm is $\tilde O(nk)$, since constructing the mapping $w_0$ takes $O(nk)$ time and running the Mettu-Plaxton algorithm on an input of size $2k$ takes $\tilde O(k^2)$ time.
Thus, we now focus on bounding the running time of Phase~II.

We can first observe that the running time of the $i^{th}$ iteration in Phase II is dominated by the total time taken to handle the calls to the algorithm $\Restr$.
In the first iteration (when $i = \ell$), we make $|Q_\ell|$ many calls to $\Restr$, each one on a subspace of size at most $n/|Q_\ell| + 2$ (by \Cref{lem:partitions:main}). Thus, by \Cref{thm:restr}, the time taken to handle these calls is at most
\begin{equation*}
\tilde O(1) \cdot \left(\frac{n}{|Q_\ell|} + 2\right)^2 \cdot |Q_\ell| 
\leq \tilde O(1) \cdot  \left(\frac{n}{|Q_\ell|}\right)^2 \cdot |Q_\ell| 
\leq \tilde O(1) \cdot  \frac{n^2}{|Q_\ell|} 
\leq \tilde O(1) \cdot \frac{n^2}{2^\ell} 
\leq \tilde O(nk),
\end{equation*}
where the first inequality follows from the fact that $n/|Q_\ell| \geq 1$, the third from \Cref{lem:partitions:main}, and the fourth since $2^\ell \geq n/k$.
It follows that the time taken to handle these calls to $\Restr$ is $\tilde O(nk)$.
For each subsequent iteration (when $0 \leq i < \ell$), we make $|Q_i|$ many calls to $\Restr$,
each one on a subspace $(X,w,d)$ of size at most $n/|Q_i| + 2$ (by \Cref{lem:partitions:main}), where the solution is restricted to the set $X \cap V_{i+1}$, which has size at most
$ |X \cap V_{i+1}| = |S_{X_1} \cup S_{X_2}| \leq 4k, $
where $\calP(X) = \{X_1, X_2\}$ and $S_{X_1}$ and $S_{X_2}$ are computed in the previous iteration.
It follows from \Cref{thm:restr} that the time taken to handle these calls is at most
$ \tilde O(1) \cdot (n/|Q_i| + 2) \cdot 4k \cdot |Q_i| \leq \tilde O(nk)$.
It follows that the total time taken to handle these calls to $\Restr$ during the $i^{th}$ iteration of Phase~II is $\tilde O(nk)$.
Hence, the total time spent handling calls to $\Restr$ is $\ell \cdot \tilde O(nk) = \tilde O(nk)$. The running time of our algorithm follows.

\section{Our Lower Bound for Deterministic $k$-Median}\label{sec:lower}

In this section, we prove the following theorem.

\begin{theorem}\label{thm:lower bound}
    For every $\delta \geq 1$, any deterministic algorithm for the $k$-median problem that has a running time of $O(kn\delta)$ on a metric space of size $n$ has an approximation ratio of
    $$ \Omega \! \left( \frac{\log n}{\log\log n + \log k + \log \delta} \right). $$
\end{theorem}
\Cref{thm:main:lower} follows from \Cref{thm:lower bound} by setting $\delta = \tilde O(1)$.

\subsection{The Proof Strategy}

Our proof of \Cref{thm:lower bound} is a modification and slight simplification of a lower bound given in the work of \cite{BateniEFHJMW23}, which provides lower bounds for various $k$-clustering problems in different computational models. 

Our proof uses the following approach:
Consider any deterministic algorithm $\A$ for the $k$-median problem. Given a metric space $(V,d)$ as input,
this algorithm can only access information about the metric space by querying the distance $d(x,y)$ between two points $x$ and $y$ in $V$. We design an adversary $\adv$ which takes as input a deterministic algorithm $\A$ and constructs a metric space $(V, \mathfrak d)$ on which the algorithm $\A$ has a bad approximation ratio. The adversary does this by running the algorithm $\A$ on a set of points $V$ and \emph{adaptively} answering the distance queries made by the algorithm in a specific way, where the queries made by the algorithm and the responses given by the adversary are a function of the previous queries and responses.
Throughout this process, the adversary constructs a metric $\mathfrak d$ on the point set $V$ which is consistent with the responses that it has given to the distance queries and also guarantees that the solution $S$ output by $\A$ at the end of this process has a bad approximation ratio compared to the optimal solution in $(V, \mathfrak d)$. Since the algorithm $\A$ is deterministic, its output when run on the metric space $(V, \mathfrak d)$ is the same as the solution $S$ that it outputs during this process.

\subsection{The Adversary $\adv$}

The adversary $\adv$ begins by creating a set of $n$ points $V$, which it feeds to an instance of $\A$ as its input.\footnote{We remark that the algorithm $\A$ is \emph{not} being given a metric space as input, since there is no metric associated with the points in $V$ at this point.} Whenever the algorithm $\A$ attempts to query the distance between two points, the adversary determines the response to the query using the strategy that we describe below.
We begin by describing the notation that we use throughout the rest of this section.

\medskip
\noindent
\textbf{Notation.}
Throughout this section, we use parameters $\delta > 1$ and $M := 10k\delta \log n$.
The parameter $\delta$ is chosen such that the query complexity of the deterministic algorithm is at most $nk\delta$.
Given a weighted graph $H$ and two nodes $u$ and $v$ of $H$, we denote the weight of the edge $(u,v)$ by $w(u,v)$ and denote the weight of the shortest path between $u$ and $v$ in $H$ by $\dist_H(u,v)$.

\medskip
\noindent
\textbf{The Graph $G$.}
The adversary $\adv$ maintains a simple, undirected graph $G$ which it uses to keep track of the queries that have already been made. The graph $G$ has $n + 1$ nodes: one special node $g^\star$ and $n$ nodes $v_x$ corresponding to each $x \in V$.
At any point in time, each node in $G$ has a \emph{status} which is either \textit{open} or \textit{closed}.
All of the nodes are initially open except $g^\star$.
Initially, the graph $G$ consists of $n$ edges of weight $\log_M n$ between $g^\star$ and each of the other nodes $v_x$ in $G$.
We note that the point $g^\star$ ensures that the distance between any two nodes in $G$ is always at most $2\log_M n$.

\medskip
\noindent
\textbf{The Auxiliary Graph $\widehat G$.}
At any point in time, we denote by $\widehat G$ the graph derived from $G$ by adding edges of weight $1$ between each pair of open nodes in $G$.
For instance, the graph $\widehat G$ initially consists of a clique of size $n$ made out of the nodes $\{v_x\}_{x \in V}$, all of whose edges have weight $1$, together with the node $g^\star$ and edges of weight $\log_M n$ between $g^\star$ and the nodes $\{v_x\}_{x \in V}$ in the clique.

\subsubsection*{Handling a Query}
We now describe how the adversary $\adv$ handles a query $\langle x,y \rangle$ and updates the graph $G$.
Depending on the status of nodes $v_x$ and $v_y$, $\adv$ does one of the following.

\medskip
\noindent
\textbf{Case 1.}
If there already exists an edge between the nodes $v_x$ and $v_y$ in $G$ (which means the distance between $x$ and $y$ is already fixed), the adversary returns the weight $w(v_x, v_y)$ as the distance between $x$ and $y$.

\medskip
\noindent
\textbf{Case 2.}
If both of $v_x$ and $v_y$ are open, $\adv$ reports the distance between $x$ and $y$ as $1$.
It then adds an edge in $G$ between $v_x$ and $v_y$ with weight $w(v_x,v_y) = 1$.
Finally, if there are any open nodes of degree at least $M$, $\adv$ sets the status of these nodes to closed.

\medskip
\noindent
\textbf{Case 3.}
If at least one of $v_x$ or $v_y$ is closed, the adversary considers the auxiliary graph $\widehat G$ (corresponding to the current graph $G$).
$\adv$ then reports the distance of $x$ and $y$ as the weighted shortest path between $v_x$ and $v_y$ in $\widehat G$, i.e.~as $\dist_{\widehat G}(v_x,v_y)$.
This shortest path contains at most one edge between two open nodes (otherwise, there would be a shortcut since the subgraph of $\widehat G$ on open nodes is a clique with all edges having weight $1$).
If such an edge $(u,v)$ between two open nodes within this shortest path exists, $\adv$ adds an edge between $u$ and $v$ in $G$ of weight $w(u,v) = 1$.
Then, $\adv$ adds an edge between $v_x$ and $v_y$ in $G$ of weight $\dist_{\widehat G}(v_x,v_y)$ (the reported distance between $x$ and $y$).
Finally, if there are any open nodes of degree at least $M$, $\adv$ sets the status of these nodes to closed.

\medskip
\noindent
\textbf{Constructing the Final Graph and Metric.}
After at most $nk\delta$ many queries, the deterministic algorithm returns a subset $S \subseteq V$ of size $k$ as its output.\footnote{We can assume w.l.o.g.~that the set $S$ contains exactly $k$ points, since we can add extra arbitrarily if $|S| \leq k$.}
Once this happens, the adversary proceeds to make some final modifications to the graph $G$.
Namely, $\adv$ pretends that the distance of each point in $S$ to every other point in $V$ is queried.
In other words, $\adv$ makes the same changes to $G$ that would occur if $\A$ had queried $\langle x, y \rangle$ for each $y \in S$ and each $x \in V$.
The order of these artificial queries is arbitrary.
We denote this final graph by $G_{f}$.

Finally, we define the metric $\mathfrak d$ on $V$ to be the weighted shortest path metric in $\widehat G_f$, i.e.~we define $\mathfrak d(x,y) := \dist_{\widehat G_f}(v_x,v_y)$ for each $x,y \in V$. The adversary then returns the metric space $(V, \mathfrak d)$, which is an instance on which $\A$ returns a solution with a bad approximation ratio.

\subsection{Analysis}

We show that the final metric $(V, \mathfrak d)$ is consistent with the answers given by the adversary to the queries made by the deterministic algorithm.
In other words, if we run $\A$ on this metric, it will return the same solution $S$. 
We defer the proof of the following lemma to \Cref{sec:consistency-proof}.

\begin{lemma}[Consistency of Metric]\label{lem:queried-weights-equal-shortest-paths}
    For each $x$ and $y$ where $\langle x,y \rangle$ is \textbf{queried}, the distance of points $x$ and $y$ in the final metric (i.e.~$\dist_{\widehat G_f}(v_x,v_y)$) equals the value returned by $\adv$ in response to the query $\langle x,y \rangle$ (i.e.~$w(v_x, v_y)$ in $\widehat G_f$).
\end{lemma}

We proceed with the analysis of the approximation ratio.
We show that the cost of $S$ as a $k$-median solution in this space is comparably higher than the cost of the optimum $k$-median solution in this space. In particular, we show that the cost of any arbitrary set of $k$ centers containing at least one point corresponding to an \textbf{open} node is small.

\begin{claim}\label{lem:number-of-neighbours-of-distance-one}
    For each $z \in S$ and $1 \leq i \leq \log_M(n)$, there are at most $M\cdot (M-1)^{i-1}$ points whose distance to $z$ is equal to $i$.
\end{claim}

We defer the proof of this claim to \Cref{sec:bound-on-neighbors}.

\begin{lemma}\label{lem:lower-bound-on-approx-of-Z}
    The cost of $S$ as a $k$-median solution is at least $(n/2) \cdot  \lfloor \log_M n \rfloor $.
\end{lemma}

\begin{proof}
    Assume $r$ is the biggest integer such that $M^r \leq n$, i.e.~$r = \lfloor \log_M n \rfloor$.
    According to \Cref{lem:number-of-neighbours-of-distance-one}, for each $1 \leq i \leq r-1$, there are at most $M(M-1)^{i-1}$ many points with distance $i$ to $z$ for any arbitrary $z \in S$.
    As a result, the total number of points whose distance to $z$ is at most $r-1$ is bounded by
    $$ 1 + M + M(M-1) + M(M-1)^2 + \cdots + M(M-1)^{r-2} \leq 2M^{r-1}. $$
    Since $|S| = k$, there are at most $2kM^{r-1}$ points whose distance to $S$ is less than or equal to $r$.
    We conclude there are at least $n - 2kM^{r-1} \geq n - 2kn/M = (1 - 2k/M) n $  points whose distance to $S$ is greater than or equal to $r$.
    Hence, the cost of $S$ is at least
    $ (1 - 2k/M)n \cdot r $.
    Note that
    $(1 - 2k/M) = (1 - 1/(5\delta\log n)) \geq 1/2$, which implies
    $ (1 - 2k/M)n \cdot r \geq  (n/2)\cdot \lfloor \log_M n \rfloor $.
\end{proof}

\begin{claim}\label{lem:number-of-closed-nodes}
    The number of closed nodes in $G_f$ is at most $(10k\delta/M)n$.
\end{claim}

\begin{proof}
    According to the building procedure of $G$, every time that $\adv$ answers a query, at most $2$ edges are added to $G$.
    This means that the total number of edges in $G_f$ is at most $2nk\delta + 2nk + n$.
    The $2nk\delta$ term is for the normal queries of the algorithm.
    The $2nk$ term is because, after the termination of the algorithm, the adversary queries at most $kn$ distances between $S$ and all other points, which adds at most $2nk$ edges to the graph in total. 
    The $n$ term is because the initial graph $G$ consists of $n$ edges. 
    Now, assume that we have $C$ many closed nodes in $G_f$.
    Since the degree of each closed node is at least $M$, we have
    $$ MC/2 \leq \ \text{total number of edges in} \ G_f \leq 2nk\delta + 2nk + n \leq 5nk\delta. $$
    As a result, $C \leq (10k\delta/M)n $.
\end{proof}

\begin{lemma}\label{lem:upper-bound-on-optimum-k-median}
    The cost of any arbitrary set of $k$ centers containing at least $1$ open node (in the final graph $G_f$) is at most $3n$.
\end{lemma}

\begin{proof}
    Let $\alpha := 10k\delta/M = 1/\log n < 1$.
    According to \Cref{lem:number-of-closed-nodes}, there are at most $\alpha n$ many closed nodes in $G$. So, at least one open node exists. Let $S^\star$ be any set of $k$ centers containing a point $x^\star$ such that $v_{x^\star}$ is an open node in $G_f$.
    The distance of any closed node to $v_{x^\star}$ is at most $2\log_M n$, since there is always the path of weight at most $2\log_M n$ between $v_{x^\star}$ and any closed node passing by $g^\star$.
    The distance between any open node and $v_{x^\star}$ is at most $1$ according to the definition of the final metric.
    Hence, the total cost of $S^\star$ which is at most the cost of assigning all of the points only to $x^\star$ is at most
    $$(\alpha n) \cdot 2\log_M n + ((1-\alpha)n) \cdot 1 = 2n \cdot \frac{\log_M n}{\log n} + (1-\alpha) \cdot n \leq 3n.\qedhere $$
\end{proof}

\noindent
Now, we are ready to prove \Cref{thm:lower bound}.
The approximation ratio of the algorithm according to \Cref{lem:lower-bound-on-approx-of-Z} and \Cref{lem:upper-bound-on-optimum-k-median} is at least
$$ \frac{(n/2) \lfloor \log_M n \rfloor}{3n} = \Omega(\log_M n). $$
Here, we assumed that $M \leq n$.
Otherwise, $\lfloor \log_M n \rfloor = 0$. Note that in the case where $M > n$, the final lower bound in \Cref{thm:lower bound} becomes a constant and the theorem is obvious in this case since the approximation ratio of every algorithm for $k$-median is at least $1$.
With some more calculations, we conclude 
$$\log_M n =\frac{\log n}{ \log \left( 10k\delta\log n\right)}  = \Omega \left(\frac{\log n}{ \log \log n + \log k + \log \delta} \right). $$

\subsection{Proof of \Cref{lem:queried-weights-equal-shortest-paths} (Consistency of The Metric)}\label{sec:consistency-proof}

    By induction on the number of queries, we show that, \textbf{at any point in time}, if there is an edge between nodes $v_x$ and $v_y$ in $G$, then 
    $w(v_x,v_y) = \dist_{\widehat G}(v_x,v_y)$.
    Assume $G_1$ is the current graph after some queries (possibly zero, at the very beginning) and that it satisfies this condition.
    Let $\langle x,y \rangle$ be the new query.
    We have the following three cases.

    \medskip
    \noindent \textbf{Case 1.}
    There is already an edge between $v_x$ and $v_y$.
    According to the strategy of $\adv$ in this case, $G_1$ is not going to change and still satisfies the property required by the lemma.

    \medskip
    \noindent \textbf{Case 2.}
    If both $v_x$ and $v_y$ are open, then an edge of weight $1$ is added to $G_1$.
    Let $G_2$ be the new graph.
    Since, $v_x$ and $v_y$ are open in $G_1$, according to the definition of $\widehat G_1$, the edge of weight $1$ between $v_x$ and $v_y$ was already in $\widehat G_1$.
    So, the edges of $\widehat G_2$ are a subset of the edges of $\widehat G_1$.
    Note that after the addition of $(v_x,v_y)$, the degree of $v_x$ or $v_y$ might become greater than or equal to $M$ and the adversary will mark them as closed in $G_2$.
    So, it is possible that $\widehat G_2$ has less edges than $\widehat G_1$ but not more edges, which concludes for each pair of arbitrary nodes $v_p$ and $v_q$, $\dist_{\widehat G_2}(v_p,v_q) \geq \dist_{\widehat  G_1}(v_p,v_q)$.
    In particular, for each pair $v_p,v_q$, such that $\langle p,q \rangle$ has been queried, we have that
    \begin{equation}\label{eq:some-inequality}
        \dist_{\widehat G_2}(v_p,v_q) \geq \dist_{\widehat  G_1}(v_p,v_q).
    \end{equation}
    Now, according to the induction hypothesis, $\dist_{\widehat G_1}(v_p,v_q) = w(v_p,v_q)$ and this edge is present in $\widehat G_2$ which can be considered as a feasible path between $v_p$ and $v_q$ in $\widehat G_2$.
    Hence, $\dist_{\widehat G_2}(v_p,v_q) \leq w(v_p,v_q) = \dist_{\widehat G_1}(v_p,v_q)$.
    Together with \Cref{eq:some-inequality}, $\dist_{\widehat G_2}(v_p,v_q) = \dist_{\widehat G_1}(v_p,v_q) = w(v_p,v_q)$ in $\widehat G_2$ as well.

    \medskip
    \noindent \textbf{Case 3.}
    At least one of $v_x$ and $v_y$ is closed.
    In this case, the adversary sets $w(v_x,v_y) = \dist_{\widehat G_1}(v_x,v_y)$.
    There might also be a new edge of weight $1$ added to $G_1$ between two open nodes.
    Let $G_2$ be the new graph.
    We have to show that for each pair of nodes $v_p$ and $v_q$ such that $\langle p,q \rangle$ has been queried, we have $w(v_p,v_q) = \dist_{\widehat G_2}(v_p,v_q)$.
    With a similar argument as the previous case, we can see that the only edge that $\widehat G_2$ might contain but $\widehat G_1$ does not contain is the new edge $(v_x,v_y)$ (In the case where the adversary also adds an edge of weight $1$ between two open nodes of $G_1$, we know that this edge is already present in $\widehat G_1$ since both its endpoints are open).
    Now, consider a shortest path $P$ between $v_p$ and $v_q$ in $\widehat G_2$.
    It is obvious that $\dist_{\widehat G_2}(v_p,v_q) \leq w(v_p,v_q)$ since $(v_p,v_q)$ itself is a valid path from $v_p$ to $v_q$ in $\widehat G_2$.
    So, it suffices to show
    \begin{equation}\label{eq:another-inequality}
        \dist_{\widehat G_2}(v_p,v_q) \geq w(v_p,v_q),
    \end{equation}
    to complete the proof.
    By the induction hypothesis, we know that $\dist_{\widehat G_1}(v_p,v_q) = w(v_p,v_q)$.
    We show that there exists a path $\tilde P$ between $v_p$ and $v_q$ in $\widehat G_1$ whose weight is exactly equal to $\dist_{\widehat G_2}(v_p,v_q)$. This implies \Cref{eq:another-inequality} since $w(v_p,v_q) = \dist_{\widehat G_1}(v_p,v_q) \leq \dist_{\widehat G_2}(v_p,v_q)$.
    Note that $\Tilde{P}$ does not need to be a valid path in $\widehat G_2$, the only condition is that the length of $\Tilde{P}$ in $\widehat G_1$ should be equal to $\dist_{\widehat G_2}(v_p,v_q)$ (the length of $P$).
    
    If $P$ does not include the new edge $(v_x,v_y)$, then $\Tilde{P} = P$ is also a valid path in $\widehat G_1$, and we are done.
    If $P$ contains the new edge $(v_x,v_y)$, we can exchange this edge $(v_x,v_y)$ with the shortest path between $v_x$ and $v_y$ in $\widehat G_1$ (which has weight $w(v_x,v_y)$ by the way this weight is constructed in response to the query $\langle x,y \rangle$).
    This gives us another path $P'$ between $v_p$ and $v_q$ in $\widehat G_1$ whose weight is exactly equal to $\dist_{\widehat G_2}(v_p,v_q)$.

\subsection{Proof of \Cref{lem:number-of-neighbours-of-distance-one}}\label{sec:bound-on-neighbors}

Before we prove \Cref{lem:number-of-neighbours-of-distance-one}, we need the following claim.

\begin{claim}\label{lem:path-of-length-1-edges}
    For each edge $(v_p,v_q)$ in $G_f$ such that $w(v_p,v_q) \leq \log_M n$, there exist a path of weight $w(v_p,v_q)$ between $v_p$ and $v_q$ in $G_f$ consisting only of edges of weight $1$.
\end{claim}

\begin{proof}
    We prove this on the current graph by induction on the number of queries. 
    Initially, there is no edge between any $v_p$ and $v_q$, so there is nothing to prove.
    Now, assume the lemma is correct for a graph $G_1$ and consider a new query $ \langle p,q \rangle$.
    If the adversary reports the distance of $p$ and $q$ as $1$, then the lemma is obvious for the new edge $(v_p, v_q)$.
    Otherwise, $w(v_p,v_q) = \dist_{\widehat G_1}(v_p,v_q)$.
    Let $G_2$ be the new graph and
    $P$ be one shortest path between $v_p$ and $v_q$ in $\widehat G_1$.
    This path $P$ contains at most one edge (say $e$) of weight $1$ between two open nodes and all of the other edges are present in $G_1$.
    Note that none of these edges are incident to $g^\star$ (since we assumed $w(v_p,v_q) \leq \log_M n$), which means we can use the induction hypothesis on these edges (except for $e$).
    For each edge $e_i \in P$ (different from $e$) of weight $w(e_i)$, by the induction hypothesis, there is a path $P_i$ of weight $w(e_i)$ consisting only of edges of weight $1$ in $G_1$.
    Finally, all of these paths $P_i$ are present in $G_2$ as well.
    Also note that the edge $e$ (if it exists) is going to be added to $G_1$ by the adversary, so $G_2$ contains $e$ itself.
    Now, we can concatenate the $P_i$ (and $e$ if exists) to get a path of weight $w(v_p,v_q)$ consisting only of edges of weight $1$ in $G_2$.
    So, the claim holds for the updated graph.
\end{proof}

Now, we proceed with the proof of \Cref{lem:number-of-neighbours-of-distance-one}.
First, we show the claim for $i = 1$.

\medskip
\noindent \textbf{Case $i=1$.} We show this case for a general $z$, not only those points that are contained in the solution $S$. So, in this case, we consider $z$ to be an arbitrary point in the space.
If $v_z$ is open, by the definition, it is obvious that $v_z$ has at most $M$ neighbors in $G_f$, and trivially the distance of $v_z$ to non-neighbor points is at least $2$.
Now, assume $v_z$ is closed. 
Consider the last time that $v_z$ was open.
So, after handling the next query, $v_z$ becomes closed.
Let $G_1$ be the graph maintained by the adversary just before handling this query.
Since $v_z$ is open in $G_1$, the degree of $v_z$ is at most $M-1$.
In the next step, at most two edges are added to $G_1$, and $v_z$ becomes closed.
So, the degree of $v_z$ is at most $M+1$.
One of the neighbors of $v_z$ is $g^\star$.
There are at most $M$ other neighbors which we denote by set $\mathcal{N}$.
Note that there is no edge between $v_z$ and any other nodes outside $\mathcal{N} + g^\star$.
After this time, since $v_z$ is closed, the distance between $v_z$ and any node $v_q$ outside $\mathcal{N}+g^\star$ is going to be at least $2$.
The reason is that the weight of the shortest path between $v_z$ and $v_q$ in $\widehat G$ that adversary considers (any time afterward) is at least $2$ (there is no edge of weight $1$ between $v_z$ and $v_q$).
As a result, the only nodes that might have distance $1$ to $v_z$ are in $\mathcal{N}$.
Since $|\mathcal{N}| \leq M$, we are done.

\medskip
\noindent \textbf{General $1 \leq i \leq \log_M n$.}
Consider $G_f$. Since the adversary queried the distance from $z$ to every other point, we know that for each node $v_p$ in the graph $(v_z,v_p)$ is an existing edge in $G_f$.
Assume the distance between $v_z$ and $v_p$ in the final metric is $i$.
According to \Cref{lem:queried-weights-equal-shortest-paths}, this distance equals $w(v_z,v_p)$ and according to \Cref{lem:path-of-length-1-edges} (since $w(v_z,v_p) = i \leq \log_M n$), there is a path $v_z = v_{p_0}, v_{p_1}, \ldots , v_{p_{i}} = v_p$ between $v_z$ and $v_p$ consisting of edges of weight $1$.
Note that nodes of the path are distinct since this is the shortest path between $v_z$ and $v_p$.
As a result, each $v_p$ corresponds to a sequence of $i+1$ pairwise distinct nodes $v_{p_0}, v_{p_1},\ldots, v_{p_i}$, such that the weight of the edge between each two consecutive nodes is $1$.
The number of such sequences is at most $M\cdot (M-1)^{i-1}$.
This is because, we have one choice for $v_{p_0}$ which is $v_z$, and $M$ choices for $v_{p_1}$ where there is an edge of weight $1$ between $v_{p_0}$ and $v_{p_1}$ (according to the proof of the above case $i=1$).
Then, for each $j \geq 2$, we have at most $M-1$ options for $v_{p_j}$ since there should exist an edge of weight $1$ between $v_{p_{j-1}}$ and $v_{p_j}$, and also $v_{p_j}$ should be different from $v_{p_{j-2}}$.
This completes the proof.

\section{Our Results for Deterministic $k$-Means}\label{sec:out-k-means}

In this section, we describe our results for the $k$-means problem, where the clustering objective defined is $\cl(S) = \sum_{x \in V}  w(x) \cdot d(x,S)^2$.

\subsection{Our Deterministic Algorithm for $k$-Means}

Our algorithm with near-optimal running time $\tilde{O}(nk)$ extends to $k$-means, giving us the following.

\begin{theorem}\label{thm:k-means}
    There is a deterministic algorithm for $k$-means that, given metric space of size $n$, computes an $O(\log^2(n/k))$-approximate solution in $\tilde{O}(nk)$ time.
\end{theorem}

Our algorithm for \Cref{thm:k-means} is identical to our $k$-median algorithm as described in \Cref{sec:our alg}.
The only difference is that we now tune everything with the objective function $\sum_{x \in V}  w(x) \cdot d(x,S)^2$ instead of $\sum_{x \in V}  w(x) \cdot d(x,S)$.
The rest of this section is devoted to proving \Cref{thm:k-means}.

\subsubsection{Projection Lemma for $k$-Means}

\begin{claim}\label{lem:triangle-means}
    For any $x,y,z \in V$ and every $0 < \epsilon < 1$, we have that
    $$ d(x,z)^2 \leq (1+\epsilon) \cdot d(x,y)^2 + (1+ 1/\epsilon) \cdot d(y,z)^2. $$
\end{claim}

\begin{proof}
    According to the Cauchy-Schwarz inequality, we have that 
\begin{align*}
    & (1+\epsilon) \cdot d(x,y)^2 + (1+ 1/\epsilon) \cdot d(y,z)^2 \\
    = &\left( (1+\epsilon) \cdot d(x,y)^2 + (1+1/\epsilon) \cdot d(y,z)^2 \right) \cdot \left( \frac{1}{1+\epsilon} + \frac{\epsilon}{1+\epsilon} \right) \\
    \geq &(d(x,y) + d(y,z))^2 \geq d(x,z)^2.
\end{align*}
\end{proof}

\begin{lemma}[Projection Lemma for $k$-Means]\label{lem:proj-k-means}
    For every $0 < \epsilon < 1$ and every subsets $A, B \subseteq V$, we have that
    $$ \cost(\pi(A,B)) \leq (1+3\epsilon) \cdot \cost(B) + (4+2/\epsilon) \cdot \cost(A). $$
\end{lemma}

\begin{proof}
Let $C$ denote $\proj(A,B)$. Let $x \in V$ and let $y^\star$ and $y$ be the closest points to $x$ in $A$ and $B$ respectively.
Let $y'$ be the closest point to $y^\star$ in $C$.
Then we have that
\begin{align*}
d(x,C)^2 &\leq d(x, y')^2 
\\
&\leq (1+\epsilon) \cdot d(y', y^\star)^2 + (1+1/\epsilon) \cdot d(y^\star, x)^2 \\
&\leq (1+\epsilon) \cdot d(y, y^\star)^2 + (1+1/\epsilon) \cdot d(y^\star, x)^2 \\
&\leq (1+\epsilon) \cdot \left( (1+\epsilon) \cdot d(y,x)^2 + (1+ 1/\epsilon) \cdot d(x,y^\star)^2 \right)  + (1+1/\epsilon) \cdot d(y^\star, x)^2  \\
&\leq (1+3\epsilon) \cdot d(x,y)^2 + (4+2/\epsilon) \cdot d(x, y^\star)^2 \\
&\leq 
(1+3\epsilon) \cdot d(x,B)^2 + (4+2/\epsilon) \cdot d(x, A)^2,
\end{align*}
These inequalities follow from $0 < \epsilon < 1$, \Cref{lem:triangle-means}, and the definitions of $y,y'$ and $y^\star$.
Hence,
\begin{align*}
  \cost(C) &= \sum_{x \in V} w(x) \cdot d(x,C)^2 \\
  &\leq \sum_{x \in V} w(x) \cdot \left( (1+3\epsilon) \cdot d(x,B)^2 + (4+2/\epsilon) \cdot d(x, A)^2 \right) \\
  &= (1+3\epsilon) \cdot \cost(B) + (4+ 2/\epsilon) \cdot \cost(A).  
\end{align*}
\end{proof}

\begin{corollary}\label{cor:improper-means}
If $Q$ is a partitioning of $V$, then
$$ \sum_{X \in Q} \OPT_k(X) \leq O(1) \cdot \OPT_k(V) . $$
\end{corollary}

\begin{proof}
    Assume that $S^\star$ is an optimal $k$-means solution on $V$.
    By considering the projection $\pi(S^\star,X)$ on every $X \in Q$, according to \Cref{lem:proj-k-means} for $\epsilon = 1/2$, we conclude that
    \begin{align*}
      \sum_{X \in Q} \OPT_k(X) &\leq \sum_{X \in Q} \cost(\pi(S^\star,X), X) \\ 
      &\leq \sum_{X \in Q} \left( (1+3/2) \cdot \cost(X,X) + (4+4) \cdot \cost (S^\star, X) \right) \\ 
      &= 8 \cdot \cost(S^\star, V) = O(1) \cdot
    \OPT_k(V).  
    \end{align*}
\end{proof}

\subsubsection{Restricted Reverse Greedy for $k$-Means}

\begin{theorem}[Analogy to \Cref{thm:restr}]\label{thm:rev-greedy-means}
    Assume $(V,w,d)$ is a metric space, and $X \subseteq V$.
    If $S_X$ is the output of $\Restr_{2k}$ running on the metric space $(X,w,d)$ while restricting the output to be a subset of $R \subseteq X$, where $|R| \leq 4k$, then for any arbitrary $0< \epsilon < 1/6$, we have that
    $$ \cost(S_X, X) \leq (1 + O(\epsilon)) \cdot \cost(R, X) + O(1/\epsilon) \cdot \OPT_k(X). $$
\end{theorem}

Assume that we run $\Restr_{2k}$ on the metric space $(X,w,d)$ while restricting the output to be a subset of $R \subseteq X$, where $|R| = m \leq 4k$.
If $m \leq 2k$, it is obvious that the output is $S = R$ without incurring any additional cost.
Suppose that we achieve nested subsets $S_{2k} \subseteq S_{2k+1} \subseteq \cdots \subseteq S_m$, where $|S_i| = i$ for each $i \in [2k, m]$ (in the case that $m \leq 2k-1$, we just define $S_{2k} := R$ as the output of the reverse greedy).
For simplicity, in \Cref{claim:greedy:proof-means} and \Cref{claim:increase-cost-means}, we abbreviate $\cost(S,X)$ by $\cost(S)$.

\begin{claim}[Analogy to \Cref{claim:greedy:proof}]\label{claim:greedy:proof-means}
    For all subsets $A \subseteq B \subseteq X$, we have that
    $$ \sum_{y \in B \setminus A} \left( \cost(B-y) - \cost(B) \right) \leq \cost(A) - \cost(B). $$
\end{claim}

\begin{proof}
    This claim follows directly from the \Cref{claim:greedy:proof} by changing the objective function.
\end{proof}

\begin{claim}[Analogy to \Cref{lem:greedy:main}]\label{claim:increase-cost-means}
    For each $i \in [2k+1, m]$ and every $0 < \epsilon < 1$, we have that
    $$\cost(S_{i-1}) \leq  \left(1+\frac{3\epsilon}{k}\right) \cdot \cost(S_i) + \frac{4+2/\epsilon}{k} \cdot \OPT_k(X).$$
\end{claim}

\begin{proof}
        Let $S^\star$ denote an optimal solution to the $k$-means problem in the metric space $(V,w,d)$.
        We denote by $S'_i$ the projection $\proj(S^\star, S_i)$ of the optimal solution $S^\star$ onto the set $S_i$. It follows that
\begin{align*}
    \cost(S_{i-1}) - \cost(S_i) &\leq \min_{y \in S_i \setminus S'_i} \left(\cost(S_i - y) - \cost(S_i) \right)\\
    &\leq \frac{1}{|S_i \setminus S'_i|} \cdot \sum_{y \in S_i \setminus S'_i} \left(\cost(S_i - y) - \cost(S_i) \right)\\
    &\leq \frac{1}{k} \cdot \sum_{y \in S_i \setminus S'_i} \left(\cost(S_i - y) - \cost(S_i) \right)\\
    &\leq \frac{1}{k} \cdot \left(\cost(S'_i) - \cost(S_i) \right)\\
    &\leq \frac{1}{k} \cdot \left( 3\epsilon \cdot \cost(S_i) + (4+2/\epsilon) \cdot \cost(S^\star) \right) \\
    &= \frac{3\epsilon}{k} \cdot \cost(S_i) + \frac{4+2/\epsilon}{k} \cdot \OPT_k(X).
\end{align*}
The first line follows directly from how the algorithm chooses which point to remove from $S_i$.
The second line follows from the fact that the minimum value within a set of real numbers is upper-bounded by its average. The third line follows from the fact that $|S_i \setminus S'_i| \geq |S_i| - |S'_i| \geq i - k \geq (2k+1) - k \geq k$. 
The fourth line follows from \Cref{claim:greedy:proof-means}.
Finally, the fifth line follows from \Cref{lem:proj-k-means}, which implies that $\cost(S'_i) \leq (1+3\epsilon) \cdot \cost(S_i) + (4+2/\epsilon) \cdot \cost(S^\star)$.
Rearranging the inequality completes the proof.
\end{proof}

\paragraph{Proof of \Cref{thm:rev-greedy-means}.}
If $m := |R| \leq 2k$, we obviously have that $S_X = R$ and the claim becomes trivial.
Now, assume that $m \geq 2k+1$.
According to \Cref{claim:increase-cost-means}, by a simple induction on $i \in [2k+1, m]$, we can show that
\begin{align*}
    \cost(S_{i-1}, X) \leq (1+3\epsilon/k)^{m-i+1} \cdot \cost(S_i, X) + \left( \sum_{j=0}^{m-i} (1+3\epsilon/k)^j \right) \cdot \frac{4+2/\epsilon}{k} \cdot \OPT_k(X). 
\end{align*}
This concludes
\begin{align*}
\cost(S_X,X) &= \cost(S_{2k},X) \\
&\leq \left(1+\frac{3\epsilon}{k}\right)^{m-2k} \cdot \cost(S_m,X) + \left( \sum_{j=0}^{m-2k-1} (1+3\epsilon/k)^j \right) \cdot \frac{4+2/\epsilon}{k} \cdot \OPT_k(X) \\
& = \left(1+\frac{3\epsilon}{k}\right)^{m-2k} \cdot \cost(R,X) + \frac{(1+3\epsilon/k)^{m-2k} - 1}{(1+3\epsilon/k) - 1} \cdot \frac{4+2/\epsilon}{k} \cdot \OPT_k(X) \\
&\leq  \left(1+\frac{3\epsilon}{k}\right)^{2k} \cdot \cost(R,X) + \frac{(1+3\epsilon/k)^{2k} - 1}{(1+3\epsilon/k) - 1} \cdot \frac{4+2/\epsilon}{k} \cdot \OPT_k(X) \\
&\leq (1 + O(\epsilon)) \cdot \cost (R,X) + \frac{(1 +  O(\epsilon)) - 1}{3\epsilon/k} \cdot \frac{3/\epsilon}{k} \cdot \OPT_k(X) \\
&\leq (1 + O(\epsilon)) \cdot \cost (R,X) + O(1/\epsilon) \cdot \OPT_k(X)
\end{align*}
The above inequalities follow from $0 < \epsilon < 1/6$ and $m \leq 4k$.

\subsubsection{Our Algorithm for $k$-Means}

The algorithm is identical to our $k$-median algorithm described in \Cref{sec:our alg}.
Here, we provide the main steps of the analysis that are analogous to those of our $k$-median algorithm.

\begin{claim}[Analogy to \Cref{lem:apx:1}]\label{claim:combine-guarantee}
For any set $X \in \bigcup_{i=0}^{\ell - 1} Q_i$ and arbitrary $0 < \epsilon < 1/6$, we have that
$$ \cost(S_X, X) \leq (1+O(\epsilon)) \cdot \sum_{X' \in \mathcal P(X)} \cost(S_{X'}, X') + O(1/\epsilon) \cdot \OPT_k(X). $$
\end{claim}

\begin{proof}
This trivially follows from \Cref{thm:rev-greedy-means} for $R := \bigcup_{X' \in \mathcal P(X)} S_{X'}$.
Note that $\cost(R,X) \leq \sum_{X' \in \mathcal P(X)} \cost(S_{X'}, X')$.
\end{proof}

\begin{claim}[Analogy to \Cref{lem:apx:2}]\label{claim:induction-cost-V0}
    For any $i \in [0, \ell]$ and any $0<\epsilon<1/6$, we have that
    $$ \cost(V_0, V) \leq (1+O(\epsilon))^i \cdot\sum_{X \in Q_i} \cost(S_{X}, X) + \left(\sum_{j=0}^{i-1} (1+O(\epsilon))^j\right) \cdot O(1/\epsilon) \cdot \OPT_k(V). $$
\end{claim}

\begin{proof}
    We prove this claim by induction.
    Note that the base case where $i = 0$ holds trivially.
    Now, suppose that the claim holds for some $i - 1 \in [0, \ell - 1]$. 
    Then we have that
    \begin{align*}
        \cost(V_0, V) &\leq (1+O(\epsilon))^{i-1} \cdot \sum_{X \in Q_{i-1}} \cost(S_{X}, X) +  \left(\sum_{j=0}^{i-2} (1+O(\epsilon))^j\right) \cdot O(1/\epsilon) \cdot \OPT_k(V).
    \end{align*}
    According to \Cref{claim:combine-guarantee}, we can bound $\sum_{X \in Q_{i-1}} \cost(S_{X}, X)$ as follows, which completes the induction step.
    \begin{align*}
        \sum_{X \in Q_{i-1}} \cost(S_{X}, X)
        &\leq \sum_{X \in Q_{i-1}} \left( (1+O(\epsilon)) \cdot \sum_{X' \in \mathcal P(X)} \cost(S_{X'}, X')
        + O(1/\epsilon) \cdot \OPT_k(X) \right) \\
        &= (1+O(\epsilon)) \cdot \sum_{X \in Q_{i-1}} \sum_{X' \in \mathcal P(X)} \cost(S_{X'}, X') + O(1/\epsilon) \cdot \sum_{X \in Q_{i-1}} \OPT_k(X)\\
        & \leq(1+O(\epsilon)) \cdot \sum_{X \in Q_{i}} \cost(S_{X}, X) + O(1/\epsilon) \cdot \OPT_k(V).
    \end{align*}
    The last inequality follows from \Cref{cor:improper-means} since $Q_{i-1}$ is a partitioning of $V$.
\end{proof}

\begin{lemma}\label{lem:bicriteria-means}
    If $\epsilon = \Theta(1/\log (n/k))$, then we have
    $\cost(V_0,V) \leq O(\log^2(n/k)) \cdot \OPT_k(V)$.
\end{lemma}

\begin{proof}
By \Cref{thm:rev-greedy-means}, it follows that, for any $X \in Q_\ell$,
\begin{equation*}
    \cost(S_X, X) \leq (1+O(\epsilon)) \cdot \cost(X,X) + O(1/\epsilon)\cdot \OPT_k(X) = O(1/\epsilon)\cdot \OPT_k(X).
\end{equation*}
this concludes $\sum_{X \in Q_\ell} \cost(S_X,X) \leq O(1/\epsilon) \cdot \OPT_k(V)$ by \Cref{cor:improper-means}.
Finally, according to \Cref{claim:induction-cost-V0} for $i = \ell = \lceil \log_2(n/k) \rceil$, we have that
\begin{align*}
    \cost(V_0, V) &\leq (1+O(\epsilon))^\ell \cdot\sum_{X \in Q_{\ell}} \cost(S_{X}, X) + \left(\sum_{j=0}^{\ell-1} (1+O(\epsilon))^j\right) \cdot O(1/\epsilon) \cdot \OPT_k(V) \\
    &\leq (1+O(\epsilon))^\ell \cdot O(1/\epsilon) \cdot \OPT_k(V) + \frac{(1+O(\epsilon))^\ell - 1}{(1+O(\epsilon)) - 1} \cdot O(1/\epsilon) \cdot \OPT_k(V) \\
    &\leq O(1/\epsilon) \cdot \OPT_k(V) + O(\ell) \cdot O(1/\epsilon) \cdot \OPT_k(V) \\
    &= O(\log^2(n/k)) \cdot \OPT_k(V).
\end{align*}
The above inequalities follow since $\epsilon = \Theta(1/\log(n/k))$ and $\ell = \lceil \log_2(n/k) \rceil = \Theta(\log(n/k))$.
\end{proof}

By \Cref{lem:bicriteria-means}, we get that $V_0$ is a $O(\log^2(n/k))$-bicriteria approximation of size $2k$. Similar to the extraction technique of \cite{focs/GuhaMMO00} (the analogous version of \Cref{lem:bicri=sparsifier} for the $k$-means problem), we can compute an exact solution to the $k$-means problem from a bicriteria approximation while only incurring constant loss in the approximation ratio, it follows that the solution $S$ constructed in Phase~III is a $O(\log^2(n/k))$-approximation and has size at most $k$.

\subsection{Our Lower Bound for Deterministic $k$-Means}

Our lower bound for deterministic $k$-median (\Cref{thm:lower bound}) extends immediately to deterministic $k$-means. In particular, we get the following theorem.

\begin{theorem}\label{thm:lower bound kmeans}
    For every $\delta \geq 1$, any deterministic algorithm for the $k$-means problem that has a running time of $O(kn\delta)$ on a metric space of size $n$ has an approximation ratio of
    $$ \Omega \! \left( \left( \frac{\log n}{\log\log n + \log k + \log \delta} \right)^2 \right). $$
\end{theorem}

\begin{proof}
    By changing the values of the parameters in the analysis for $k$-median, we can obtain our lower bound for $k$-means.
    Let $M = 10k\delta \log^2n$. Then the cost of the solution $S$ returned by the algorithm is at least $(1- 2k/M) \cdot n \cdot r^2 \geq (n/2) \cdot r^2$, where $r = \lfloor \log_M n \rfloor$.
    This follows from the same argument as the proof of \Cref{lem:lower-bound-on-approx-of-Z}, except that using the $k$-means objective instead of the $k$-median objective gives us the $r^2$ term.
    By the same argument as the proof of \Cref{lem:upper-bound-on-optimum-k-median}, the cost of the optimum solution is at most
    $$ (10k\delta/M) \cdot (2\log_M n)^2 + (1 - 10k\delta/M) \cdot n \cdot 1 = 4 \cdot \frac{\log^2_M n}{\log^2 n} + n \leq 5n . $$
    Thus, the approximation ratio of the algorithm is at least
    $$ \frac{(n/2) \lfloor \log_M n \rfloor^2}{5n} = \Omega(\log^2_M n) = \Omega \left( \left( \frac{\log n}{\log \log n + \log k + \log \delta} \right)^2 \right).\footnote{Similar to $k$-median, we assume that $M \leq n$. Otherwise, the theorem is trivial.}\qedhere $$
\end{proof}

\bibliographystyle{alpha}
\bibliography{references}

@inproceedings{HLRW24,
  author       = {Monika Henzinger and
                  Jason Li and
                  Satish Rao and
                  Di Wang},
  title        = {Deterministic Near-Linear Time Minimum Cut in Weighted Graphs},
  booktitle    = {Proceedings of the 2024 {ACM-SIAM} Symposium on Discrete Algorithms,
                  {SODA} 2024, Alexandria, VA, USA, January 7-10, 2024},
  pages        = {3089--3139},
  publisher    = {{SIAM}},
  year         = {2024},
}

@inproceedings{MN20,
  author       = {Sagnik Mukhopadhyay and
                  Danupon Nanongkai},
  title        = {Weighted min-cut: sequential, cut-query, and streaming algorithms},
  booktitle    = {Proceedings of the 52nd Annual {ACM} {SIGACT} Symposium on Theory
                  of Computing, {STOC} 2020, Chicago, IL, USA, June 22-26, 2020},
  pages        = {496--509},
  publisher    = {{ACM}},
  year         = {2020},
}

@inproceedings{BateniEFHJMW23,
  author       = {MohammadHossein Bateni and
                  Hossein Esfandiari and
                  Hendrik Fichtenberger and
                  Monika Henzinger and
                  Rajesh Jayaram and
                  Vahab Mirrokni and
                  Andreas Wiese},
  title        = {Optimal Fully Dynamic \emph{k}-Center Clustering for Adaptive and
                  Oblivious Adversaries},
  booktitle    = {Proceedings of the 2023 {ACM-SIAM} Symposium on Discrete Algorithms (SODA)},
  pages        = {2677--2727},
  publisher    = {{SIAM}},
  year         = {2023},
}

@article{AryaGKMMP04,
  author       = {Vijay Arya and
                  Naveen Garg and
                  Rohit Khandekar and
                  Adam Meyerson and
                  Kamesh Munagala and
                  Vinayaka Pandit},
  title        = {Local Search Heuristics for k-Median and Facility Location Problems},
  journal      = {{SIAM} J. Comput.},
  volume       = {33},
  number       = {3},
  pages        = {544--562},
  year         = {2004},
}

@inproceedings{MettuP02,
  author       = {Ramgopal R. Mettu and
                  C. Greg Plaxton},
  title        = {Optimal Time Bounds for Approximate Clustering},
  booktitle    = {Proceedings of the 18th Conference in Uncertainty in Artificial
                  Intelligence (UAI)},
  pages        = {344--351},
  year         = {2002},
}

@inproceedings{MettuP00,
  author       = {Ramgopal R. Mettu and
                  C. Greg Plaxton},
  title        = {The Online Median Problem},
  booktitle    = {41st Annual Symposium on Foundations of Computer Science (FOCS)},
  pages        = {339--348},
  year         = {2000},
}

@inproceedings{HenzingerK20,
  author       = {M.~Henzinger and
                  S.~Kale},
  title        = {Fully-Dynamic Coresets},
  booktitle    = {ESA},
  year         = {2020},
}

@inproceedings{ourneurips2023,
  author       = {Sayan Bhattacharya and
                  Mart{\'{\i}}n Costa and
                  Silvio Lattanzi and
                  Nikos Parotsidis},
  title        = {Fully Dynamic $k$-Clustering in $\tilde {O}(k)$ Update Time},
  booktitle    = {Advances in Neural Information Processing Systems 36: Annual Conference
                  on Neural Information Processing Systems 2023, NeurIPS 2023, New Orleans,
                  LA, USA, December 10 - 16, 2023},
  year         = {2023},
}

@inproceedings{focs/GuhaMMO00,
  author       = {Sudipto Guha and
                  Nina Mishra and
                  Rajeev Motwani and
                  Liadan O'Callaghan},
  title        = {Clustering Data Streams},
  booktitle    = {41st Annual Symposium on Foundations of Computer Science, {FOCS} 2000,
                  12-14 November 2000, Redondo Beach, California, {USA}},
  pages        = {359--366},
  publisher    = {{IEEE} Computer Society},
  year         = {2000},
}

@inproceedings{nips/Cohen-AddadHPSS19,
  author       = {Vincent Cohen{-}Addad and
                  Niklas Hjuler and
                  Nikos Parotsidis and
                  David Saulpic and
                  Chris Schwiegelshohn},
  title        = {Fully Dynamic Consistent Facility Location},
  booktitle    = {Advances in Neural Information Processing Systems 32: Annual Conference
                  on Neural Information Processing Systems 2019, NeurIPS 2019, December
                  8-14, 2019, Vancouver, BC, Canada},
  pages        = {3250--3260},
  year         = {2019},
}

@article{ChrobakKY06,
  author       = {Marek Chrobak and
                  Claire Kenyon and
                  Neal E. Young},
  title        = {The reverse greedy algorithm for the metric \emph{k}-median problem},
  journal      = {Inf. Process. Lett.},
  volume       = {97},
  number       = {2},
  pages        = {68--72},
  year         = {2006},
}

@inproceedings{charikar1999constant,
  title={A constant-factor approximation algorithm for the k-median problem},
  author={Charikar, Moses and Guha, Sudipto and Tardos, {\'E}va and Shmoys, David B},
  booktitle={Proceedings of the thirty-first annual ACM symposium on Theory of computing},
  pages={1--10},
  year={1999}
}

@article{jain2001approximation,
  title={Approximation algorithms for metric facility location and k-median problems using the primal-dual schema and Lagrangian relaxation},
  author={Jain, Kamal and Vazirani, Vijay V},
  journal={Journal of the ACM (JACM)},
  volume={48},
  number={2},
  pages={274--296},
  year={2001},
  publisher={ACM New York, NY, USA}
}

@article{ahmadian2019better,
  title={Better guarantees for k-means and euclidean k-median by primal-dual algorithms},
  author={Ahmadian, Sara and Norouzi-Fard, Ashkan and Svensson, Ola and Ward, Justin},
  journal={SIAM Journal on Computing},
  volume={49},
  number={4},
  pages={FOCS17--97},
  year={2019},
  publisher={SIAM}
}

@article{byrka2017improved,
  title={An improved approximation for k-median and positive correlation in budgeted optimization},
  author={Byrka, Jaros{\l}aw and Pensyl, Thomas and Rybicki, Bartosz and Srinivasan, Aravind and Trinh, Khoa},
  journal={ACM Transactions on Algorithms (TALG)},
  volume={13},
  number={2},
  pages={1--31},
  year={2017},
  publisher={ACM New York, NY, USA}
}

@inproceedings{charikar2003better,
  title={Better streaming algorithms for clustering problems},
  author={Charikar, Moses and O'Callaghan, Liadan and Panigrahy, Rina},
  booktitle={Proceedings of the thirty-fifth annual ACM symposium on Theory of computing},
  pages={30--39},
  year={2003}
}

@article{ailon2009streaming,
  title={Streaming k-means approximation},
  author={Ailon, Nir and Jaiswal, Ragesh and Monteleoni, Claire},
  journal={Advances in neural information processing systems},
  volume={22},
  year={2009}
}

@article{tcs/Gonzalez85,
  author       = {Teofilo F. Gonzalez},
  title        = {Clustering to Minimize the Maximum Intercluster Distance},
  journal      = {Theor. Comput. Sci.},
  volume       = {38},
  pages        = {293--306},
  year         = {1985},
}

@inproceedings{focs/BCLP24,
  author       = {Sayan Bhattacharya and
                  Mart\'in Costa and
                  Naveen Garg and
                  Silvio Lattanzi and
                  Nikos Parotsidis},
  title        = {Fully Dynamic $k$-Clustering with Fast Update Time and Small Recourse},
  booktitle={65th IEEE Symposium on Foundations of Computer Science (FOCS)},
  year={2024}
}

@article{GuptaT08,
  author       = {Anupam Gupta and
                  Kanat Tangwongsan},
  title        = {Simpler Analyses of Local Search Algorithms for Facility Location},
  journal      = {CoRR},
  volume       = {abs/0809.2554},
  year         = {2008},
  url          = {http://arxiv.org/abs/0809.2554},
  eprinttype    = {arXiv},
  eprint       = {0809.2554},
}

@article{dam/HsuN79,
  author       = {Wen{-}Lian Hsu and
                  George L. Nemhauser},
  title        = {Easy and hard bottleneck location problems},
  journal      = {Discret. Appl. Math.},
  volume       = {1},
  number       = {3},
  pages        = {209--215},
  year         = {1979},
}

@article{DupreS24,
  author       = {Max Dupr{\'{e}}{~}la{~}Tour and
                  David Saulpic},
  title        = {Almost-linear Time Approximation Algorithm to Euclidean k-median and
                  k-means},
  journal      = {CoRR},
  volume       = {abs/2407.11217},
  year         = {2024},
  url          = {https://doi.org/10.48550/arXiv.2407.11217},
  doi          = {10.48550/ARXIV.2407.11217},
  eprinttype    = {arXiv},
  eprint       = {2407.11217},
}

@inproceedings{esa/TourHS24,
  author       = {Max Dupr{\'{e}}{~}la{~}Tour and
                  Monika Henzinger and
                  David Saulpic},
  title        = {Fully Dynamic k-Means Coreset in Near-Optimal Update Time},
  booktitle    = {32nd Annual European Symposium on Algorithms, {ESA} 2024},
  series       = {LIPIcs},
  volume       = {308},
  pages        = {100:1--100:16},
  publisher    = {Schloss Dagstuhl - Leibniz-Zentrum f{\"{u}}r Informatik},
  year         = {2024},
}

@inproceedings{ourstoc25,
  author       = {Sayan Bhattacharya and
                  Mart{\'{\i}}n Costa and
                  Ermiya Farokhnejad},
  title        = {Fully Dynamic $k$-Median with Near-Optimal Update Time and Recourse},
  booktitle    = {57th Annual {ACM} {SIGACT} Symposium on Theory of Computing (STOC)},
  year         = {2025},
  note={(To Appear)},
}

@inproceedings{focs/Cohen-AddadSS23,
  author       = {Vincent Cohen{-}Addad and
                  David Saulpic and
                  Chris Schwiegelshohn},
  title        = {Deterministic Clustering in High Dimensional Spaces: Sketches and
                  Approximation},
  booktitle    = {64th {IEEE} Annual Symposium on Foundations of Computer Science, {FOCS}
                  2023, Santa Cruz, CA, USA, November 6-9, 2023},
  pages        = {1105--1130},
  publisher    = {{IEEE}},
  year         = {2023},
}

@article{talg/NeimanS16,
  author       = {Ofer Neiman and
                  Shay Solomon},
  title        = {Simple Deterministic Algorithms for Fully Dynamic Maximal Matching},
  journal      = {{ACM} Trans. Algorithms},
  volume       = {12},
  number       = {1},
  pages        = {7:1--7:15},
  year         = {2016},
}

@inproceedings{stoc/AssadiCS22,
  author       = {Sepehr Assadi and
                  Andrew Chen and
                  Glenn Sun},
  editor       = {Stefano Leonardi and
                  Anupam Gupta},
  title        = {Deterministic graph coloring in the streaming model},
  booktitle    = {{STOC} '22: 54th Annual {ACM} {SIGACT} Symposium on Theory of Computing,
                  Rome, Italy, June 20 - 24, 2022},
  pages        = {261--274},
  publisher    = {{ACM}},
  year         = {2022},
}

@inproceedings{focs/HaeuplerLS24,
  author       = {Bernhard Haeupler and
                  Yaowei Long and
                  Thatchaphol Saranurak},
  title        = {Dynamic Deterministic Constant-Approximate Distance Oracles with n\({}^{\mbox{{\(\epsilon\)}}}\)
                  Worst-Case Update Time},
  booktitle    = {65th {IEEE} Annual Symposium on Foundations of Computer Science, {FOCS}
                  2024, Chicago, IL, USA, October 27-30, 2024},
  pages        = {2033--2044},
  publisher    = {{IEEE}},
  year         = {2024},
}

@inproceedings{Chang16,
  author       = {Ching{-}Lueh Chang},
  title        = {Metric 1-Median Selection: Query Complexity vs. Approximation Ratio},
  booktitle    = {Computing and Combinatorics - 22nd International Conference, {COCOON}
                  2016, Ho Chi Minh City, Vietnam, August 2-4, 2016, Proceedings},
  series       = {Lecture Notes in Computer Science},
  volume       = {9797},
  pages        = {131--142},
  publisher    = {Springer},
  year         = {2016},
}

@article{young25,
  author       = {Neal E. Young},
  title        = {An improved approximation algorithm for k-Median},
  journal      = {CoRR},
  volume       = {abs/2511.12230},
  year         = {2025},
  url          = {https://doi.org/10.48550/arXiv.2511.12230},
  doi          = {10.48550/ARXIV.2511.12230},
  eprinttype   = {arXiv},
  eprint       = {2511.12230},
}

\newpage

\appendix

\section{The Algorithm of \cite{focs/GuhaMMO00} (Proof of \Cref{thm:intro:guha})}\label{sec:guha alg}

In this section, we prove \Cref{thm:intro:guha}, which we restate below.

\begin{theorem}\label{thm:main:time:restate}
    There is a deterministic algorithm for $k$-median that, given a metric space of size $n$, computes a $\poly(\log(n/k) / \log \delta)$-approximate solution in $\tilde O(nk\delta)$ time, for any $2 \leq \delta \leq n/k$.
\end{theorem}

\subsection{Preliminaries}

In this section, for ease of notation, we consider solutions to the $k$-median problem to be mappings instead of subsets of points.
More precisely, we denote a solution to the $k$-median problem on $(V,w,d)$ by a mapping $\sigma : V \longrightarrow S$, where $S = \sigma(V)$ is the set of centres of size at most $k$, and each $x \in V$ is assigned to center $\sigma(x)$.
We denote the cost of this solution by
$\cost(\sigma,V,w) := \sum_{x \in V} w(x)d(x, \sigma(x))$.

\medskip
\noindent\textbf{The Mettu-Plaxton Algorithm.}
This algorithm uses the $\tilde O(n^2)$ time algorithm of \cite{MettuP00} as a black box, which we refer to as $\MPAlg$ for convenience.\footnote{We can also use any other $O(1)$-approximate algorithm that runs in time $\tilde O(n^2)$.} The following theorem summarizes the properties of $\MPAlg$.

\begin{theorem}[\cite{MettuP00}]\label{thm:MP alg} There exists a deterministic algorithm $\MPAlg$ that, given a metric space of size $n$, returns a $O(1)$-approximation to the $k$-median problem in at most $\tilde O(n^2)$ time.
\end{theorem}

For notational convenience, we denote the approximation ratio and the hidden polylogarithmic function in the running time of $\MPAlg$ by $\alpha$ and $A$ respectively. Thus, given a metric space of size $n$, $\MPAlg$ returns an $\alpha$-approximation in time at most $A \cdot n^2$.

\subsection{The Algorithm}

Let $(V, w, d)$ be a metric space of size $n$, $k \leq n$ be an integer and $\delta > 1$ be a parameter.
We also define values $\ell := \lceil \log_2(\log(n/k)/\log \delta) \rceil$, $\gamma := n/k$, and $q_i := \lceil \gamma^{1/2^{i}} \rceil$ for each $i \in [\ell]$, which we use to describe the algorithm.
The algorithm works in 2 \emph{phases}, which we describe below.

\medskip
\noindent
\textbf{Phase I:} In the first phase of the algorithm, we construct a sequence of partitions $Q_0,\dots, Q_\ell$ of the metric space $V$, such that the partition $Q_i$ is a \emph{refinement} of the partition $Q_{i-1}$.\footnote{i.e.~for each element $X \in Q_{i-1}$, there are elements $X_1,\dots, X_q \in Q_i$ such that $X=X_1\cup \dots \cup X_q$.} We start off by setting $Q_0 := \{V\}$. Subsequently, for each $i = 1,\dots,\ell$, we construct the partition $Q_i$ as follows:

\begin{wrapper}
    Initialize $Q_i \leftarrow \varnothing$. Then, for each $X \in Q_{i-1}$, arbitrarily partition $X$ into subsets $X_1, \dots, X_{q_i}$ such that $\left||X_{j}| - |X_{j'}|\right| \leq 1$ for each $j,j' \in [q_i]$, and add these subsets to $Q_i$.
\end{wrapper}

\medskip
\noindent
\textbf{Phase II:}
The second phase of the algorithm proceeds in \emph{iterations}, where we use the partitions $\{Q_i\}_i$ to compute the solution in a bottom-up manner.
Let $V_{\ell + 1}$ and $w_{\ell + 1}$ denote the set of points $V$ and the weight function $w$ respectively.
For each $i = \ell ,\dots, 0$, the algorithm constructs $V_i$ as follows:

\begin{wrapper}
    For each $X \in Q_i$, let $\sigma_X$ be the solution obtained by running $\MPAlg$ on the metric space $(X \cap V_{i+1},w_{i+1},d)$ and $S_X := \sigma_X(X \cap V_{i+1})$. For each center $y \in S_X$, let $w_i(y) := \sum_{x \in \sigma_{X}^{-1}(y)} w_{i+1}(x)$ be the total weight (w.r.t.~$w_{i+1}$) of the points assigned to $y$ in $X$ by $\sigma_X$. Let $V_i := \bigcup_{X \in Q_i} S_X$.
\end{wrapper}

\medskip
\noindent
\textbf{Output:}
For each $i \in [0, \ell]$, let $\sigma_i : V_{i+1} \longrightarrow V_i$ denote the mapping obtained by taking the union of the mappings $\{\sigma_X\}_{X \in Q_i}$.\footnote{Note that, since the domains of the mappings $\{\sigma_X\}_{X \in Q_i}$ partition $V_{i+1}$, their union is well defined.}
The output of the algorithm is the mapping $\sigma : V \longrightarrow V_0$, which we define as the composition
$\sigma := \sigma_{0} \circ \dots \circ \sigma_{\ell}$ of the $\sigma_i$.

\subsection{Analysis}

We now analyze the algorithm by bounding its approximation ratio and running time. We begin by proving the following lemmas, which summarize the relevant properties of the partitions constructed in Phase I of the algorithm.

\begin{lemma}\label{lem:partitions}
    For each $i \in [\ell]$, the set $Q_i$ is a partition of $V$ into $\prod_{j=1}^i q_j$ many subsets of size at most $n/|Q_i| + i$. Furthermore, $Q_i$ is a refinement of $Q_{i-1}$.
\end{lemma}

\begin{proof}
    We define $Q_0$ as $\{V\}$, which is a trivial partition of $V$. Now, suppose that this statement holds for the partition $Q_i$, where $0 \leq i < \ell$. The algorithm constructs $Q_{i+1}$ by taking each $X \in Q_i$ and further partitioning $X$ into subsets $X_1,\dots,X_{q_{i+1}}$, such that difference in the sizes of any two of these subsets is at most $1$. Clearly, the partition $Q_{i+1}$ is a refinement of the partition $Q_i$. We can also observe that the number of subsets in the partition $Q_{i+1}$ is $q_{i+1} \cdot |Q_{i}| = q_{i+1} \cdot \prod_{j=1}^i q_j = \prod_{j=1}^{i+1} q_j$.\footnote{Note that we do not necessarily guarantee that all of the sets in these partitions are non-empty.} Finally, since each subset $X \in Q_i$ has size at most $n/|Q_i| + i$, it follows that each subset in $Q_{i+1}$ has size at most
    $$ \left\lceil \frac{1}{q_{i+1}} \cdot \left( \frac{n}{|Q_i|} + i \right) \right\rceil \leq \frac{n}{q_{i+1} \cdot |Q_i|} + \frac{i}{q_{i+1}} + 1 \leq \frac{n}{|Q_{i+1}|} + (i+1).\qedhere $$
\end{proof}

\begin{claim}\label{app:lem:prod q}
    For each $i \in [\ell]$, we have that $\gamma^{1 - 1/2^i} \leq \prod_{j=1}^i q_j \leq e^i \cdot \gamma^{1 - 1/2^i}$.
\end{claim}

\begin{proof}
    Let $i \in [\ell]$. For the lower bound, we can see that
    \begin{equation}\label{eq:prod q 1}
        \prod_{j=1}^i q_j = \prod_{j=1}^i \lceil \gamma^{1/2^j} \rceil \geq \prod_{j=1}^i \gamma^{1/2^j} = \gamma^{\sum_{j=1}^i 1/2^j} = \gamma^{1 - 1/2^i}.
    \end{equation}    
    For the upper bound, we use the fact that $\lceil x \rceil \leq x + 1 = x (1 + 1/x)$ to get that
    \begin{equation}\label{eq:prod q 2}
        \prod_{j=1}^i q_j = \prod_{j=1}^i \lceil \gamma^{1/2^j} \rceil \leq \prod_{j=1}^i \gamma^{1/2^j} \cdot \left(1 + \gamma^{-1/2^j}\right) = \left(\prod_{j=1}^i \gamma^{1/2^j} \right) \cdot \left(\prod_{j=1}^i \left(1 + \gamma^{-1/2^j}\right) \right).
    \end{equation}
    It follows from \Cref{eq:prod q 1} that $\prod_{j=1}^i \gamma^{1/2^j} = \gamma^{1 - 1/2^i}$. We can also see that
    $$ \prod_{j=1}^i \left(1 + \gamma^{-1/2^j}\right) \leq \prod_{j=1}^i \exp \! \left(\gamma^{-1/2^j}\right) = \exp \! \left(\sum_{j=1}^i \gamma^{-1/2^j}\right) \leq e^i.$$
    Thus, combining these upper bounds with \Cref{eq:prod q 2}, it follows that $\prod_{j=1}^i q_j \leq e^i \cdot \gamma^{1 - 1/2^i}$.
\end{proof}

\subsubsection*{Approximation Ratio}
To analyze the approximation ratio of the algorithm, we use the following lemma of \cite{focs/GuhaMMO00}, which shows that we can use a good bicriteria approximation for $k$-median as a `sparsifier' for the underlying metric space. A proof of this lemma using the same notation as our paper can be found in \cite{focs/BCLP24}.

\begin{lemma}[Lemma 10.3, \cite{focs/BCLP24}]\label{lem:bicri=sparsifier}
    Let $(V,w,d)$ be a metric space,
    $\sigma : V \longrightarrow V'$ be a mapping such that $\cost(\sigma, V, w) \leq \beta \cdot \OPT_k(V, w)$, and define $w'(y) := \sum_{x \in \sigma^{-1}(y)} w(x)$ for all $y \in V'$. Given a mapping $\pi : V' \longrightarrow S$ such that $\cost(\pi, V', w') \leq \alpha \cdot \OPT_k(V', w')$, we have that
    $$ \cost(\pi \circ \sigma, V, w) \leq (2\alpha + (1 + 2\alpha)\beta) \cdot \OPT_k(V, w). $$
\end{lemma}

For each $i \in [0, \ell]$, let $\sigma_i'$ denote the mapping $\sigma_i \circ \dots \circ \sigma_\ell$. Note that the output of the algorithm is precisely $\sigma'_0$.
We use \Cref{lem:bicri=sparsifier} to inductively bound the approximation ratio of each $\sigma'_i$. In particular, we prove the following lemma.

\begin{lemma}\label{lem:apx rec}
    For each $0 \leq i \leq \ell$, $\cost (\sigma'_{i}, V, w) \leq (9\alpha)^{\ell + 1 - i} \cdot \OPT_k(V, w)$.
\end{lemma}

\begin{proof}
    Let $\sigma_{\ell + 1}'$ denote the identity mapping on $V$. Then, we clearly have that $\cost(\sigma'_{\ell + 1}, V, w) = 0$.
    Now, let $i \in [0,\ell]$ and suppose that
    \begin{equation}\label{eq:apx 1}
        \cost (\sigma_{i+1}', V, w) \leq (9\alpha)^{\ell - i} \cdot \OPT_k(V, w).
    \end{equation}
    Let $\sigma^\star_{i+1}$ denote an optimal solution to the $k$-median problem in the metric space $(V_{i+1}, w_{i+1}, d)$.
    We can upper bound the cost of $\sigma_i$ (in the space $(V_{i+1}, w_{i+1}, d)$) by
    \begin{align}
        \cost(\sigma_i, V_{i+1}, w_{i+1}) &= \sum_{X \in Q_i} \cost(\sigma_X, X \cap V_{i+1}, w_{i+1})\nonumber\\
        &\leq \sum_{X \in Q_i} 2\alpha \cdot \cost(\sigma^\star_{i+1}, X \cap V_{i+1}, w_{i+1})\nonumber\\
        &= 2\alpha \cdot \cost(\sigma^\star_{i+1}, V_{i+1}, w_{i+1})\nonumber\\
        &= 2\alpha \cdot \OPT_k(V_{i+1}, w_{i+1}),\label{eq:apx 2}
    \end{align}
    where the first and third lines follow from the fact that $\{X \cap V_{i+1} \mid X \in Q_{i}\}$ partitions $V_{i+1}$ and
    the second line from $\sigma_X$ being an $\alpha$-approximate solution in the metric space $(X \cap V_{i+1}, w_{i+1}, d)$. We note that the extra factor of $2$ in the third line follows from the fact that $\sigma_{i+1}^\star(X \cap V_{i+1})$ might contain points that are not in $X \cap V_{i+1}$.
        
    Now, since we have that $\sigma_{i}' = \sigma_i \circ \sigma'_{i+1}$, we can apply \Cref{lem:bicri=sparsifier} using the upper bounds on $\cost (\sigma_{i+1}', V, w)$ and $\cost(\sigma_i, V_{i+1}, w_{i+1})$ given in \Cref{eq:apx 1,eq:apx 2} to get that
    $$\cost(\sigma'_{i}, V, w) \leq (4\alpha + (1 + 4\alpha) \cdot (9\alpha)^{\ell - i}) \cdot \OPT_k(V,w) \leq (9\alpha)^{\ell + 1 - i} \cdot \OPT_k(V,w). \qedhere$$
\end{proof}

\noindent
It follows from \Cref{lem:apx rec} by setting $i=0$ that 
$$\cost(V_0, V, w) = 2^{O(\ell)} \cdot \OPT_k(V,w) = \poly (\log(n/k)/\log \delta) \cdot \OPT_k(V,w).$$

\subsubsection*{Running Time}

The running time of Phase I of the algorithm is $O(n \ell) = \tilde O(n)$, since it takes $O(n)$ time to construct each partition $Q_i$ given the partition $Q_{i-1}$. Thus, we now focus on bounding the running time of Phase II.

We can first observe that the running time of the $i^{th}$ iteration in Phase II is dominated by the total time taken to handle the calls to the algorithm $\MPAlg$.
In the first iteration (when $i = \ell$), we make $|Q_\ell|$ many calls to $\MPAlg$, each one on a subspace of size at most $n/|Q_\ell| + \ell$ (by \Cref{lem:partitions}). Thus, by \Cref{thm:MP alg}, the time taken to handle these calls is at most
\begin{equation}\label{eq:1}
A \cdot \left(\frac{n}{|Q_\ell|} + \ell\right)^2 \cdot |Q_\ell| \leq A \ell^2 \cdot  \left(\frac{n}{|Q_\ell|}\right)^2 \cdot |Q_\ell| = A \ell^2 \cdot  \frac{n^2}{|Q_\ell|} = A \ell^2 \cdot \frac{n^2}{\prod_{j=1}^\ell q_j} \leq A \ell^2 \cdot \frac{n^2}{\gamma^{1 - 1/2^\ell}}, 
\end{equation}
where the first inequality follows from the fact that $n/|Q_\ell| \geq 1$ and $\ell \geq 1$, the last equality follows from \Cref{lem:partitions}, and the last inequality follows from \Cref{app:lem:prod q}.
We can now upper bound the RHS of \Cref{eq:1} by
\begin{equation}\label{eq:3}
A \ell^2 \cdot \frac{n^2}{\gamma^{1 - 1/2^\ell}} = A \ell^2 \cdot n^2 \cdot  \frac{k}{n} \cdot \gamma^{1/2^\ell} = A \ell^2 \cdot nk \cdot \gamma^{2^{-\ell}} \leq A \ell^2 \cdot nk \cdot \delta.  
\end{equation}
Thus, the time taken to handle these calls to $\MPAlg$ is $\tilde O(nk\delta)$.
For each subsequent iteration (when $0 \leq i < \ell$), we make $|Q_i|$ many calls to $\MPAlg$, 
each one on a subspace $(X \cap V_{i+1},w_{i+1},d)$ of size at most $q_{i+1} k$ since
$ |X \cap V_{i+1}| = | \bigcup_{j=1}^{q_{i+1}} S_{X_j}| \leq q_{i+1} k$,
where $X_1, \ldots, X_{q_{i+1}}$ are the subsets that $X$ is partitioned into, and each $S_{X_j}$ is the solution computed on the subspace $(X_j \cap V_{i+2}, w_{i+2},d)$ in the previous iteration.
It follows that the time taken to handle these calls is at most
\begin{align*}
A \cdot (q_{i+1} k)^2 \cdot |Q_i| &= A \cdot (q_{i+1} k)^2 \cdot \prod_{j=1}^i q_j = A \cdot k^2 \cdot q_{i+1} \cdot \prod_{j=1}^{i+1} q_j \\ &\leq A \cdot k^2 \cdot 2\gamma^{1/2^{i+1}} \cdot e^i \gamma^{1 - 1/2^{i+1}} = 2Ae^i \cdot k^2 \cdot \gamma =  2A e^i \cdot nk.    
\end{align*}
where the first equality follows from \Cref{lem:partitions} and the first inequality follows from \Cref{app:lem:prod q} and the fact that $q_{i+1} \leq 2\gamma^{1/2^{i+1}}$. Since $\ell = O(\log \log n)$, we get that $e^i \leq e^\ell = \tilde O(1)$ and it follows that the total time taken to handle these calls to $\MPAlg$ is $\tilde O(nk)$.
Consequently, the running time of the algorithm is $\tilde O(nk\delta) + \ell \cdot \tilde O(nk) = \tilde O(nk\delta)$.

\subsection{Extension to $k$-Means}\label{sec:guha:kmeans}

It is straightforward to extend this algorithm to the $k$-means problem, where the clustering objective is $\sum_{x \in V} w(x) \cdot d(x,S)^2$ instead of $\sum_{x \in V} w(x) \cdot d(x,S)$. In particular, we get the following theorem.

\begin{theorem}
    There is a deterministic algorithm for $k$-means that, given a metric space of size $n$, computes a $\poly(\log(n/k) / \log \delta)$-approximate solution in $\tilde O(nk\delta)$ time, for any $2 \leq \delta \leq n/k$.
\end{theorem}

We define the \emph{normalized $k$-means} objective as $\cost_2(S) := \left(\sum_{x \in V}  w(x) \cdot d(x,S)^2\right)^{1/2}$. As pointed out by \cite{focs/BCLP24}, for technical reasons, it is easier to work with this notion of normalized cost. By observing that a solution $S$ is an $\alpha^2$-approximation to $k$-means if and only if it is an $\alpha$-approximation to normalized $k$-means, we can assume w.l.o.g.~that we are working with the normalized objective function.

Since the Mettu-Plaxton algorithm \cite{MettuP02} can also be used to give a $O(1)$-approximation to the normalized $k$-means problem, this algorithm works for normalized $k$-means without any modification. Thus, the running time guarantees extend immediately. Furthermore, \cite{focs/BCLP24} show that the exact statement of \Cref{lem:bicri=sparsifier} holds for the normalized $k$-means objective, i.e.~replacing $\cost$ with $\cost_2$. Using this lemma, it is easy to see that the analysis of the approximation also extends with no modifications.

\end{document}